\documentclass[11pt,oneside,letterpaper]{article}
\usepackage[mode=buildnew]{standalone}
\usepackage{eurosym}
\usepackage{graphicx,amsmath,amsfonts}
\usepackage{setspace}
\usepackage{amssymb}
\usepackage{dsfont}
\usepackage{amsmath}
\usepackage{amsfonts}
\usepackage{ifthen}
\usepackage{mathtools}
\usepackage{graphicx}
\usepackage{enumitem}
\usepackage{color}
\usepackage{framed}
\usepackage[margin=1in]{geometry}
\usepackage{natbib}
\usepackage[dvipsnames]{xcolor}
\usepackage[colorlinks=true, urlcolor=Mahogany, linkcolor=Mahogany, citecolor=Mahogany]{hyperref}
\usepackage{fp}
\usepackage[capitalize]{cleveref}
\usepackage{amsthm}
\usepackage{caption}
\usepackage{subcaption}
\usepackage{palatino}
\usepackage{cuted,balance}
\usepackage{amssymb}
\usepackage{soul}
\usepackage{graphicx,amsmath,amsfonts}
\usepackage{float}
\usepackage{bbm}

\usepackage[ruled]{algorithm2e}

\DeclareMathOperator*{\argmax}{arg\,max}
\newtheorem{assumption}{Assumption}

\newtheorem{corollary}{Corollary}
\newtheorem{definition}{Definition}
\newtheorem{proposition}{Proposition}
\newtheorem{theorem}{Theorem}
\newtheorem{lemma}{Lemma}
\newtheorem{claim}{Claim}
\newtheorem{fact}{Fact}
\newtheorem{remark}{Remark}

\allowdisplaybreaks

\let\originaleqref\eqref
\renewcommand{\eqref}{\originaleqref}
\usepackage{bm}
\newcommand{\bb}{\bm{b}}
\newcommand{\bx}{\bm{x}}
\newcommand{\bp}{\bm{p}}
\newcommand{\bv}{\bm{v}}
\newcommand{\bh}{\bm{h}}

\usepackage{tikz}
\usetikzlibrary{positioning}
\usepackage{pgfplots}
\usetikzlibrary{intersections, pgfplots.fillbetween}

\begin{document}

\title{Fair Allocation in Dynamic Mechanism Design}
\author{Alireza Fallah \thanks{University of California, Berkeley} \and 
Michael I. Jordan$^*$ \and Annie Ulichney$^*$}

\date{October 2024}
\maketitle

\sloppy

\begin{abstract}
We consider a dynamic mechanism design problem where an auctioneer sells an indivisible good to groups of buyers in every round, for a total of $T$ rounds. The auctioneer aims to maximize their discounted overall revenue while adhering to a fairness constraint that guarantees a minimum average allocation for each group. We begin by studying the static case ($T=1$) and establish that the optimal mechanism involves two types of subsidization: one that increases the overall probability of allocation to all buyers, and another that favors the groups which otherwise have a lower probability of winning the item. We then extend our results to the dynamic case by characterizing a set of recursive functions that determine the optimal allocation and payments in each round. Notably, our results establish that in the dynamic case, the seller, on the one hand, commits to a participation bonus to incentivize truth-telling, and on the other hand, charges an entry fee for every round. Moreover, the optimal allocation once more involves subsidization, which its extent depends on the difference in future utilities for both the seller and buyers when allocating the item to one group versus the others. Finally, we present an approximation scheme to solve the recursive equations and determine an approximately optimal and fair allocation efficiently.
\end{abstract}

\section{Introduction}
Auctions play a pivotal role in allocation decisions across various domains, serving as a structured methodology for determining the distribution of resources based on bids. Housing, government contracts, electromagnetic spectrum rights, and advertising slots are a few of the many domains where allocations are determined by auctions. In many of the real-world applications of auction-based allocations, the resource in question is both valuable and limited, which raises concerns about fairness. 

For instance, housing auctions are often governed by policies addressing the needs of low- and moderate-income households, such as Affordable Housing Quotas in the United Kingdom to tax credits and vouchers in the United States \citep{gallent2000planning, keightley2014introduction}. Government procurement contracts for goods and services in the United States are subject to the United States Small Business Act which outlines target contracting rates across categories (e.g., woman-owned businesses, veteran-owned businesses, historically underutilized zones) \citep{cravero2017socially}. The United States Federal Communications Commission drives initiatives ensuring , for instance, that small and regional companies maintain reasonable access to connectivity \citep{rachfal2019broadband}. 

The importance of incorporating a consideration of fairness in important real-world use cases motivates recent work that extends classical work on auction design emphasizing economic efficiency to ensure outcomes do not disproportionately favor some bidders over others. Integrating fairness into auction designs is, however, a source of significant complexity, given that buyers often have different valuations, which can lead to strategic bidding behavior that may skew fairness. 

In this paper, we focus on studying the question of fairness in a dynamic mechanism design setting. We consider a setting in which a seller allocates an indivisible item through an auction at each of $T$ rounds, to two groups of buyers, with each group comprising $n$ buyers. At each round, the buyers’ values are realized from an underlying distribution (which may differ between the two groups), and the buyers then submit their bids. The seller's goal is to decide on an allocation rule as a function of the submitted bids that maximize their total discounted revenue with a discount factor $\delta$, while ensuring that each group receives a minimum allocation, or more specifically, that group $i$'s average discounted allocation is greater than or equal to some $\alpha_i$.


This fairness constraint enforces a proportional notion of fairness regarding the number of items each group wins. There are several advantages to adopting such a fairness notion. First, it is typically more straightforward from a policy perspective compared to, for instance, focusing on the total value won by each group. For example, most regulations concerning fair housing allocation emphasize the percentage of housing allocated to low-income groups as a measure of success \citep{mccarty2014overview}. Moreover, a buyer’s value in private-value auctions is not solely determined by the intrinsic value of the item but also depends on the financial constraints of the buyers or asymmetric information about the item’s value. Therefore, a proportional fairness constraint is a suitable choice in cases where one group’s values are systematically lower than another’s. For example, in the context of housing allocation, low-income individuals may assign lower monetary value to homes compared to wealthier individuals, as they are more concerned with basic shelter than with investment or resale value. In such scenarios, the allocation ratio is a more effective approach to mitigate allocation inequality compared to using buyers’ value.

We begin by considering the static case with $T=1$. In the absence of constraints, it is well known that the optimal allocation in this case is the second-price auction with a reserve price, which allocates the item to the buyer with the highest bid, conditioned on the bid being above a certain reserve price. Here, we establish that the optimal allocation under the fairness constraint has two additional features: first, the optimal fair mechanism subsidizes one group that would otherwise not meet a target fairness constraint. Second, the optimal allocation may increase the chance of allocating the item to both groups in general by reducing the reserve price (again disproportionately in favor of one group).

We next focus on the case of general $T > 1$. Here, we characterize the optimal allocation through a set of recursive functions. At any round $t$, we define $\text{residual minimum allocation}$, which, roughly speaking, given the allocation so far, updates the fairness constraint as if we were to start at this round. We establish that, given the residual minimum allocation, we may have to allocate the item to one of the two groups to be on track to satisfy the fairness constraint overall, and in this case, we would simply allocate it to the highest bidder in the group and charge them at the second-highest bid within their group.  That said, when we have the option to allocate to both groups, the optimal allocation rule is similar to the static case at a high level: when allocating an item to a group, it is assigned to the buyer with the highest bid. Second, the decision regarding which group receives the item takes the form of a subsidy. 
However, determining the exact amount of this subsidy, as well as the payment functions, is more challenging in the dynamic case because buyers in each group may consider underreporting their value in the current round to increase their chances of winning the item in future rounds, especially if they expect their value to be higher in those future rounds.

In particular, we establish that, as the above reasoning suggests, there is a utility associated with not winning the item. This utility is the difference between a buyer’s expected utility in future rounds if their group does not win the item this round, compared to the scenario where their group wins the item at the current round. We show that, at the optimal allocation, and to incentivize buyers to report their values truthfully, the seller commits to paying them a participation bonus if their group wins the item, and this bonus is equal to the aforementioned net utility of not winning the item. On the other hand, the seller charges them an entry fee, which is equal to the expected bonus payment that they would miss if they did not participate. 
Finally, the amount of the subsidy in optimal allocation also depends on the above net utility of not winning the item as well as the seller's future utility from allocating the item to one group or the other.

In summary, our results show that the fairness constraint in the dynamic case introduces two new features in contrast to the classical auction: a subsidy term in the allocation rule and two new payment transfers, specifically the entry fee paid by all participants to the seller and the participation bonus paid by the seller to the winning group. We further examine the impact of these new payment terms and show that, while the participation bonus the seller must pay could exceed the entry fee, its excess decreases as the number of buyers grows, meaning that for large auctions, the entry fee offsets the participation bonus.

As this overview suggests, finding the optimal allocation requires solving a set of recursive functions, and as we will see, the computational complexity may accordingly grow exponentially with the number of rounds. Our next contribution is to provide an efficient approximation scheme. In particular, we establish that, for any $\varepsilon > 0$, we can find an allocation that guarantees the same utility to the seller and a fairness constraint of $(1-\varepsilon)(\alpha_i-\varepsilon)$ for group $i$. This approximation algorithm requires $\mathcal{O}(\varepsilon^{-1/(1-\delta))})$ calls to an oracle that computes integrals. We further provide a constant approximation scheme which requires $\text{Poly}(1/(1-\delta))$ calls to the oracle, which is more applicable to scenarios where $\delta$ is close to one. 

 We conclude our paper by demonstrating that all our results and insights in both the static and dynamic case extend to the case with more than two groups. We also present a set of a numerical experiments that illustrates the impact of the fairness constraint on the utility of both the seller and the buyers.

\subsection{Related Work}
Our work is related to a growing body of literature on fair resource allocation and fairness in mechanism design. Here we highlight particularly relevant segments of the rich set of works that address normative, methodological, and algorithmic aspects of fair division. For a broader review of the literature, we refer to an extended set of related works \Cref{sec:related_work} as well as several recent surveys ~\citep[see, e.g.,][]{bouveret2016fair, walsh2020fair, amanatidis2023fair}.

\paragraph{Notions of Fairness} Normative questions surrounding allocative fairness are longstanding, where foundational works established prominent notions of fairness for the allocation of divisible goods, notably envy-freeness \citep{foley1966resource} and proportionality \citep{steinhaus1948problem}.  However, classical notions of fairness often cannot be feasibly implemented for indivisible goods \citep{procaccia2014fair}. The complexities of the indivisible setting necessitate new notions and methodologies of fair allocation that introduce randomness \citep{babaioff2022best} as well as those that are approximations or relaxations of earlier concepts \citep{budish2011combinatorial, caragiannis2023new, conitzer2017fair, amanatidis2018comparing}. Where classical notions of fairness assumed agents are entitled to an equal allocate share, our work joins those that generalize to allow for weighted or asymmetric notions of fairness ~\citep[see, e.g.,][]{babaioff2023fair, chakraborty2022weighted}. 

\paragraph{Truthful Fair Allocation Mechanisms} Where foundational works in fair division largely concern the non-strategic allocation setting, our work joins those addressing the setting of strategic users and incentive compatibility. 
The closest setting to ours is that of \citet{pai2012auction}, who examine a similar notion of fairness in the single-round (static) case, but only for one group. We depart from their work in our consideration of the dynamic case, multiple groups, and an an approximation scheme. Additionally, even in the static case, we consider a minimum allocation constraint for all groups, which introduces a second type of subsidization. 

The works \cite{lipton2004approximately}, \cite{babichenko2023fair} and \cite{caragiannis2009low} consider truthful fair allocations with respect to variants of envy-freeness, \cite{sinha2015mechanism} take the social welfare maximization approach, and \cite{amanatidis2016truthful},  \cite{amanatidis2017truthful}, \cite{amanatidis2023allocating}, and  \cite{cole2013mechanism} consider allocation to strategic users under variants of proportional notions of fairness. The works \cite{babaioff2022fair} and \cite{babaioff2024share} explore more generally the feasibility of allocations across classes of fairness constraints in the strategic setting. Our work considers the same strategic setting and provides the efficient mechanism satisfying a given possibly-asymmetric proportional notion of fairness (if one exists).

\paragraph{Dynamic Fairness} Our work is also a departure from the dominating static focus of much of the existing fair division literature; we present a mechanism that satisfies a fairness constraint over numerous rounds. We build upon foundational works in dynamic mechanism design surveyed in \cite{vohra2012dynamic} and \cite{bergemann2019dynamic} by additionally considering fairness in the dynamic strategic setting. Other works that consider allocative fairness in a dynamic setting vary in their characterization of what changes from round to round. Where some works suppose that resources are fixed while buyers change in time \citep{sinclair2022sequential, vardi2022dynamic}, our work joins those that consider sequentially allocating a stream of items \citep[]{aleksandrov2015online, aleksandrov2017pure, zeng2020fairness, gkatzelis2021fair, bogomolnaia2022fair}. However, our mechanism is suited for a (possibly asymmetric) proportional rather than envy-free notion of fairness, and we also allow for valuations that are not necessarily binary-valued. Other works in the dynamic setting consider fairness amidst uncertainty in, for instance, supply \citep[]{bateni2022fair, benade2024fair}. Our work addresses the same sequential setting where only buyer values are stochastic, however, our mechanism incentivizes truthfulness rather than assuming it. Similarly, several works consider a dynamic notion of fairness under uncertain demand with a focus on effectively meeting demand \citep[]{lien2014sequential, elzayn2019fair, donahue2020fairness, manshadi2021fair}. These works share our ex ante notion of fairness, but depart from our setting by specifying a decision rule that is committed to ex ante relative to the realization of buyer values each round.

\paragraph{Fairness and Efficiency} Our work also connects with literature on market-based redistribution and the relationship of fairness and efficiency. This line of research explores optimal mechanisms that balance efficient allocation and redistribution in response to inequality, with applications in, for instance, housing markets, vaccine distribution, and nonprofit operations \citep{dworczak2021redistribution, akbarpour2024economic, ma2023fairness, alexsandrov2015foodbank}.

\paragraph{Empirical Perspectives} A separate set of related works presents empirical studies of the impact on efficiency of subsidizing targeted groups of bidders in procurement auctions \citep{krasnokutskaya2011bid,athey2013subsidy}. These works complement ours by pointing to bid subsidies as the cost-optimal mechanism to enforce fairness.

\section{Model} \label{sec:model}
We consider a monopolist seller that interacts with two groups of buyers over $T$ rounds, with $n$ buyers in each group.\footnote{Here, for simplicity in notation, we assume that the sizes of two groups are equal, each being $n$. Our analysis and results are, however, applicable to the general case with unequal group sizes.} At each round $t \in [T]$, the seller conducts an auction to sell a single indivisible item. The private value of buyer $k \in [n]$ from group $i \in [2]$ (or simply, buyer $(i,k)$) for this item is denoted by $v_{i,k}^t$, which is drawn from the publicly known and full-support continuous distribution with density $f_{i}^t: [\underline{v}_i^t, \bar{v}_i^t] \to \mathbb{R}$. We denote its cumulative distribution function by $F_{i}^t(\cdot)$. We assume the buyers' values are independent within and between rounds.

\paragraph{Timeline of the auction:} At each time $t$, buyers' private values are realized, and then each buyer $(i,k)$ submits a bid $b_{i,k}^t$. We denote the vector of bids at round $t$ by $\bb^t$. The seller then conducts the auction and decides on the allocation $\bx^t(\bb^t, \bh^t)$ where $\bh^t$ denotes the public history of allocations up to (but not including) time $t$, i.e., which buyers received the items in the previous $t-1$ rounds. In particular, $x_{i,k}^t(\bb^t, \bh^t) \in \{0,1\}$ denotes whether the $k$th buyer in group $i$ gets the item in round $t$. Notice that since we allocate a single item at each round, we should have $\sum_{i=1}^2 \sum_{k=1}^n x_{i,k}^t(\bb^t, \bh^t) \leq 1$ for every $t \in [T]$. Finally, buyer $(i,k)$ makes the payment $p_{i,k}^t(\bb^t, \bh^t)$ to the seller. The vector of payments at time $t$ is denoted by $\bp^t(\bb^t, \bh^t)$. Notice that the mechanism is determined by the allocation and payment functions $(\bx^{1:T}, \bp^{1:T})$. 

\paragraph{Utility functions:} The utility function of buyer $(i,k)$ at time $t$ is given by 
\begin{equation} \label{eqn:buyer_k_utility_t}
\mathcal{U}_{i,k}^t(v_{i,k}^t; \bb^t, \bh^t) := x_{i,k}^t(\bb^t, \bh^t) v_{i,k}^t - p_{i,k}^t(\bb^t, \bh^t).    
\end{equation}
The buyer $(i,k)$'s overall utility over $T$ rounds is given by
\begin{equation} \label{eqn:buyer_k_utility}
\mathcal{U}_{i,k}(\bv_{i,k}^{1:T}; \bb^{1:T}, \bh^{1:T}) := \sum_{t=1}^T \delta^{t-1} \mathcal{U}_{i,k}^t(v_{i,k}^t; \bb^t, \bh^t),   
\end{equation}
where $\delta \in (0,1]$ is the discount factor. 
The seller's utility over all rounds is given by
\begin{equation} \label{eqn:seller_utility}
\mathcal{U}_{0}(\bb^{1:T}, \bh^{1:T}) := \sum_{t=1}^T \delta^{t-1} \sum_{i=1}^2 \sum_{k=1}^n p_{i,k}^t(\bb^t, \bh^t).    
\end{equation}
Our goal is to identify mechanisms that maximize the seller's expected utility, subject to certain fairness constraints ensuring a minimum allocation is guaranteed to each group. We next formally define these fairness constraints.  
\paragraph{The fairness constraint:} The seller, acting as the auctioneer, aims to ensure that a minimum allocation is guaranteed to each group in this dynamic setting. More specifically, let $\alpha_i$ represent the minimum (discounted) average allocation promised to group $i$, with the condition that $\alpha_1 + \alpha_2 \leq 1$. Then, at each round $t$, the seller ensures that the discounted average of the items allocated to group $i$ thus far, combined with the expected allocation in the remaining rounds, is at least $\alpha_i$:
\begin{equation} \label{eqn:fairness_constraint}
\frac{1-\delta}{1-\delta^T} \left( \sum_{\tau=1}^{t-1} \delta^{\tau-1} \sum_{k=1}^n x_{i,k}^\tau + \mathbb{E} \left [ \sum_{\tau=t}^T \delta^{\tau-1} \sum_{k=1}^n x_{i,k}^\tau \right ] \right ) \geq \alpha_i,   
\end{equation}
where the expectation is taken over the realization of the private values at time $t$ and the future rounds. We call the auction \emph{infeasible} at round $t$ if it is not possible to satisfy the fairness constraint over the remaining rounds. 

It is straightforward to verify that if the condition \eqref{eqn:fairness_constraint} holds for $t=T$, then it would hold for any $t \leq T$. However, as we will elaborate later, this condition determines the set of feasible allocations for any $t \leq T-1$. For instance, suppose we have $T=4$ and $\alpha_1 = \alpha_2 = {1}/{3}$. Then, if we allocate the item to the first group in the first three rounds, we cannot satisfy the condition \eqref{eqn:fairness_constraint} for the second group in the last round. 

Notice that this condition is ex ante with respect to the allocation at the last round. We find this condition more compatible compared to an ex post fairness in a setting with indivisible goods and for small $T$. For instance, with $T=2$ and $\delta = 1$, the ex post average allocation to each group is either $0$, ${1}/{2}$, or $1$, and hence, setting $\alpha_1 = \alpha_2 = \alpha > 0$ for any $\alpha \leq 1/2$ is essentially no different from setting $\alpha_1 = \alpha_2 = {1}/{2}$ in that case. However, as $T$ grows, this constraint becomes increasingly similar to an ex post fairness constraint, differing only in the very last round, which has a diminishing weight of $\delta^T$. We will further discuss the ex post fairness constraint in Section \ref{sec:dynamic} when we characterize the optimal dynamic allocation.
\paragraph{Direct truthful mechanisms:} The dynamic revelation principle states that, without loss of generality, we can focus on direct, truthful mechanisms in which buyers bidding truthfully is a Nash equilibrium \citep{myerson1986multistage}. It is important to note that the dynamic revelation principle requires truthful reporting only on the equilibrium path; that is, assuming all buyers have been truthful in the past. However, in our setting, this would imply truthful reporting under any history, even if a buyer has previously been dishonest (and we are off the equilibrium path). To see why this is the case, note that the utility of each buyer depends on their current (private) value and all past bids (but not on the past true values). Thus, when it is optimal for the buyer to bid truthfully when past reports have been truthful, it remains optimal for them to bid truthfully even if some buyers have lied in the past. This argument generally holds true for Markovian environments. See \cite{pavan2014dynamic} for a detailed discussion on this topic. We also refer readers to \cite{sugaya2021revelation} for a broader treatment of the conditions under which the revelation principle holds in dynamic mechanism design and for different solution concepts. 

In this paper, we use the \textit{periodic ex post incentive compatibility} (see \cite{bergemann2019dynamic} for a discussion on this definition and its comparison with the weaker notion of Bayesian incentive compatibility in dynamic mechanism design). More specifically, for any mechanism $(\bx^{1:T}, \bp^{1:T})$, the following Markovian decision problem finds the best bid for buyer $k$ from group $i$, assuming that all other buyers are reporting their values truthfully:
\begin{equation} \label{eqn:recursive_utility}
U_{i,k}^t(\bv^t, \bh^t) := \max_{b_{i,k}^t} \left \{
\mathcal{U}_{i,k}^t \left (v_{i,k}^t; b_{i,k}^t, \bv_{-(i,k)}^t, \bh^t \right )
+ \delta~ \mathbb{E} \left [ U_{i,k}^{t+1}(\bv^{t+1}, \bh^{t+1}) \Big \vert \bh^t
\right ] \right \},
\end{equation}
with the convention that $U_{i,k}^{T+1}(\cdot, \cdot) := 0$.
This is a recursive formula in which the buyer identifies the optimal bid by maximizing the sum of their utility at the current moment and the maximum of what they could achieve in the future, while they assume that all other buyers are bidding truthfully. A mechanism is called ex post incentive compatible if the maximum is achieved through truthful bidding. The formal definition is provided below.
\begin{definition}
A mechanism $(\bx^{1:T}, \bp^{1:T})$ is \emph{ex post incentive compatible} (EPIC) if for all $i, k, t, \bv^t,$ and $\bh^t$:
\begin{equation} \label{eqn:EPIC}
v_{i,k}^t \in \argmax_{b_{i,k}^t} \left \{
\mathcal{U}_{i,k}^t \left (v_{i,k}^t; b_{i,k}^t, \bv_{-(i,k)}^t, \bh^t \right )
+ \delta~ \mathbb{E} \left [ U_{i,k}^{t+1}(\bv^{t+1}, \bh^{t+1}) \Big \vert \bh^t
\right ] \right \}.    
\end{equation}
\end{definition}
Note that the EPIC mechanism can be identified through backward induction, starting at the last round and progressing back to the first. By doing so, at time $t$, we see that maximum utility in future rounds is achieved through truthful bidding. Consequently, the term $U_{i,k}^{t+1}(\bv^{t+1}, \bh^{t+1})$ could be replaced with the utility assuming that everyone reports truthfully. In other words, to identify the EPIC mechanism, it is sufficient to find a mechanism that satisfies the following condition:
\begin{equation} \label{eqn:EPIC_2}
v_{i,k}^t \in \argmax_{b_{i,k}^t} \left \{
\mathcal{U}_{i,k}^t \left (v_{i,k}^t; b_{i,k}^t, \bv_{-(i,k)}^t, \bh^t \right )
+ \delta~ \mathbb{E} \left [ \sum_{\tau=t+1}^T \mathcal{U}_{i,k}^\tau(v_{i,k}^\tau; \bv^\tau, \bh^\tau) \Big \vert \bh^t
\right ] \right \}.     
\end{equation}
Finally, we state the individual rationality assumption or the participation condition. This assumption ensures that, at each round, each buyer (knowing only their value) would weakly prefer to participate in the auction. In other words, each buyer's expected utility is weakly higher if they participate rather than choosing their outside option, which is skipping that round.
\begin{definition}
A mechanism $(\bx^{1:T}, \bp^{1:T})$ is \emph{individually rational} (IR) if for all $i, k, t, v_{i,k}^t,$ and $\bh^t$: 
\begin{equation} \label{eqn:IR}
\mathbb{E}\left[ U_{i,k}^t(\bv^t, \bh^t) \right] 
\geq \delta~ \mathbb{E}\left [ U_{i,k}^{t+1}(\bv^{t+1}, (\bh')^{t+1}) \Big \vert \bh^t
\right ] ,
\end{equation}
where $(\bh')^{t+1}$ denotes the history under which buyer $(i,k)$ has not participated at round $t$, and the expectations are taken over $\bv^t_{-(i,k)}$ and $\bv^{t+1:T}$.\footnote{In the case of $n=1$, a buyer's non-participation could lead to the infeasibility of the auction as the item cannot be allocated to their group in that round. To avoid such special cases, in the case of $n=1$, we assume that we may still allocate the item to their group even if they do not participate, but this potential allocation does not count towards their outside option utility.}
\end{definition}
The left-hand side of \eqref{eqn:IR} captures the expected utility of buyer $(i,k)$ at time $t$ and thereafter, assuming their participation in round $t$. The right-hand side of \eqref{eqn:IR},. on the other hand, shows the expected utility of buyer $(i,k)$ when they skip the $t$-th round (and receive no utility in the $t$-th round).

Throughout our analysis, we adopt the following regularity assumption on the distribution of values that is common in the mechanism design literature. This assumption applies to a variety of distributions, including uniform, exponential, and normal distributions. 
\begin{assumption} \label{assumption:regularity}
For any $i$ and $t$, the distribution $F_i^t$ is regular, i.e., the virtual value function
\begin{equation} \label{eqn:virtual_value}
\phi_i^t(v) := v - \frac{1-F_i^t(v)}{f_i^t(v)}
\end{equation}
is increasing in $v$. 
\end{assumption}
While we work with this assumption in the main body of the paper, in Appendix \ref{proof:ironing}, we demonstrate that this assumption can be relaxed through an \emph{ironing} procedure. For the rest of the paper, our goal is to find an EPIC and IR mechanism $(\bx^{1:T}, \bp^{1:T})$ that maximizes the expected seller utility subject to the fairness constraint \eqref{eqn:fairness_constraint}. We start our analysis by discussing the static case, i.e., $T=1$, which will subsequently be used as a building block in our treatment of the dynamic setting.
\section{The Static Case}
To simplify the notation, we drop the dependence of the parameters on time $t$ in this section. We first revisit the following well-known result for EPIC and IR mechanisms in the static scenario (see Chapters 9 and 13 of \cite{AGT} for the proof):
\begin{proposition}[\cite{myerson1981optimal}] \label{proposition:static_EPIC}
A mechanism $(\bx, \bp)$ is EPIC if and only if for every $i$ and $k$ we have that
\begin{itemize}
\item the allocation function $x_{i,k}(v, \bv_{-(i,k)})$ is weakly increasing in $v$, and
\item the payment function is given by 
\begin{equation} \label{eqn:payment_static}
p_{i,k}(v, \bv_{-(i,k)}) = v~x_{i,k}(v, \bv_{-(i,k)}) - \int_{\underline{v}_i}^v x_{i,k}(z, \bv_{-(i,k)}) dz + c(\bv_{-(i,k)}),   
\end{equation}
\end{itemize}
with $c(\bv_{-(i,k)})=0$ for IR mechanisms. Moreover, the seller's expected utility for any EPIC and IR mechanism is given by
\begin{equation}\label{eqn:seller_utility_static} 
\mathbb{E}_{\bv} \left [ \sum_{i \in [2],k \in [n]} \phi_{i}(v_{i,k}) x_{i,k}(\bv) \right]. 
\end{equation}
\end{proposition}
\paragraph{Static fairness constraint:}
In the static setting where a single item is auctioned in one round, to satisfy the ex ante fairness constraint, the seller ensures that the expected total allocation for group $i$ is as least $\alpha_i$. In other words, in the static setting, our fairness constraint (\ref{eqn:fairness_constraint}) reduces to
\begin{equation} \label{eqn:static_fairness_constraint}
\mathbb{E} \left [ \sum_{k=1}^n x_{i,k} \right ] \geq \alpha_i,   
\end{equation}
where the expectation is taken over a single realization of the private values. Notice that the definition of our fairness constraint allows for the simultaneous control of two notions of allocative fairness. The first is to control the relative probabilities of allocation between groups one and two which we will see takes the form of a subsidy-like modification to the unconstrained mechanism. The second is to control the overall rates of allocation which we will see effectively reduces the seller's reserve price. The relative values of $\alpha_i$ and the allocations of the unconstrained efficient mechanism determine which of these notions of fairness are relevant in a given setting. 

It is well known that in the absence of the fairness constraint, and under Assumption \ref{assumption:regularity}, the item is allocated to the buyer who has the highest virtual value, provided that at least one buyer's virtual value is nonnegative. Conversely, the seller does not allocate the item if all virtual values are negative. This mechanism is known as the second-price auction (or Vickrey auction) with reserve pricing, as detailed in \cite{AGT}.
This allocation is depicted in Figure \ref{fig:allocation_static}. 

\begin{figure}
\centering
\begin{subfigure}{.5\textwidth}
  \centering
  \resizebox{\linewidth}{!}{\begin{tikzpicture}[>=stealth]

\newcommand\leneps{1.25}
\newcommand\xmax{4}
\newcommand\xmin{-2}
\newcommand\ymax{4}
\newcommand\ymin{-2}

    \draw[->] (\xmin,0) -- (\xmax,0) node[right]{$\phi_1(v_1)$};
    \draw[->] (0,\ymin) -- (0,\ymax) node[above]{$\phi_2(v_2)$};

    \draw[-, white, name path = boundary1] (\xmin,\ymin) -- (\xmin,-\leneps);
    \draw[-, white, name path = boundary2] (\xmin,-\leneps)--(\xmin,0);
    \draw[-, white, name path = boundary3] (\xmin,0)--(\xmin,\ymax);
    \draw[-, white, name path = boundary4] (\xmin,\ymax)--(\xmax,\ymax);
    \draw[-, white, name path = boundary5] (\xmax,\ymax)--(\xmax,\ymax-\leneps);
    \draw[-, white, name path = boundary6] (\xmax,\ymax-\leneps)--(\xmax,\ymin)--(0,\ymin);
    \draw[-, white, name path = boundary7] (0,\ymin)--(\xmin,\ymin);

    \draw[dashed, ,->, black, name path = grayline1] (0,0) -- (4,4) node[right]{$\phi_2(v_2) = \phi_1(v_1)$};
    \draw[dashed, ,->, black, name path = grayline2] (0,0) -- (-2,0);

    \tikzfillbetween[of=grayline2 and boundary4]{blue, opacity=0.2};
    \tikzfillbetween[of=grayline1 and boundary6]{yellow, opacity=0.2};

\end{tikzpicture}}
  \caption{The optimal allocation without the fairness \\ constraint.}
  \label{fig:allocation_static}
\end{subfigure}%
\begin{subfigure}{.5\textwidth}
  \centering
  \resizebox{\linewidth}{!}{\begin{tikzpicture}[>=stealth]
\def\leneps{1.25}
\def\xmax{4}
\def\xmin{-4}
\def\ymax{4}
\def\ymin{-4}
\def\etaone{1}
\pgfmathsetmacro\etatwo{\etaone + \leneps}
\definecolor{group1_color1}{HTML}{ECD9ED}

    \draw[->] (\xmin,0) -- (\xmax,0) node[right]{$\phi_1(v_1)$};
    \draw[->] (0,\ymin) -- (0,\ymax) node[above]{$\phi_2(v_2)$};

    \draw[-, white, name path = boundary1] (\xmin,\ymin) -- (\xmin,-\leneps);
    \draw[-, white, name path = boundary2] (\xmin,-\leneps)--(\xmin,0);
    \draw[-, white, name path = boundary3] (\xmin,0)--(\xmin,\ymax);
    \draw[-, white, name path = boundary4] (\xmin,\ymax)--(\xmax,\ymax);
    \draw[-, white, name path = boundary5] (\xmax,\ymax)--(\xmax,\ymax-\leneps);
    \draw[-, white, name path = boundary6] (\xmax,\ymax-\leneps)--(\xmax,\ymin)--(0,\ymin);
    \draw[-, white, name path = boundary7] (0,\ymin)--(\xmin,\ymin);

    \filldraw (-\etaone,0) circle (1pt);
    \node[above] at (-\etaone, 0) {$-\eta_1$};
    \filldraw (0, -\etatwo) circle (1pt);
    \node[right] at (0, -\etatwo) {$-\eta_2$};
    \filldraw (0, -\leneps) circle (1pt);
    \node[right] at (0, -\leneps) {$-\gamma$};

    \draw[dashed, ,->, black, name path = x_eq_y_pos] (0,0) -- (4,4) node[right]{$\phi_2(v_2) = \phi_1(v_1)$};
    \draw[dashed, ,->, black, name path = y_eq_0_neg] (0,0) -- (-\ymin,0);

    \draw[dashed, ,->, black, name path = x_eq_y_mineps_pos] (-\etaone,-\etatwo) -- (4,4-\leneps) node[right]{$\phi_1(v_1) = \phi_2(v_2) + \gamma$};

    \draw[dashed, ,->, lightgray, name path = x_eq_y_mineps_pos_neg] (-\etaone,-\etatwo) -- (-3,-3-\leneps);

    \draw[dashed, ,->, black, name path = y_eq_min_eps_neg] (0,-\etatwo) -- (\xmin,-\etatwo);

    \draw[dashed, ,->, black, name path = x_eq_neg_delta1] (-\etaone,0) -- (-\etaone,-4);

    \draw[dashed, ,->, black, name path = y_eq_neg_eps] (0, -\leneps) -- (-\xmax,-\leneps);

    \fill[green, opacity=0.2] (0, 0) -- (\xmax, \ymax) -- (\xmax, \ymax - \leneps) -- (0, -\leneps) -- (-\xmax, -\leneps) -- (-\xmax, 0) -- cycle;

    \fill[blue, opacity=0.2] (0, 0) -- (-\xmax, 0)-- (-\xmax, \ymax) -- (0, \ymax) -- (\xmax, \ymax) -- (0, 0) -- cycle;

    \fill[orange, opacity=0.2] (-\etaone, -\etatwo) -- (-\etaone, -\ymax) -- (0, -\ymax)-- (0, -\leneps) -- cycle;

    \fill[yellow, opacity=0.2](0, -\leneps) -- (\xmax, \ymax - \leneps) -- (\xmax, -\ymax) -- (0, -\ymax)-- cycle;

    \fill[green, opacity=0.2](0, -\leneps) -- (-\xmax, -\leneps) -- (-\xmax, -\etatwo) -- (-\etaone, -\etatwo) -- cycle;

\end{tikzpicture}}
  \caption{The optimal allocation under the fairness \\ constraint.}
  \label{fig:fair_allocation_static}
\end{subfigure}
\caption{Optimal allocation in the static case.}
\label{fig:Static}
\end{figure}
We next present the main result of this section which establishes the optimal allocation under the fairness constraint.
\begin{theorem}\label{thrm:fair_allocation_static} 
Suppose Assumption \ref{assumption:regularity} holds.\footnote{In \cref{proof:ironing}, we demonstrate that the result can be extended to the case that the regularity assumption is weakened using the ironing technique.} For any $i \in [2]$, let $v_i := \max_{k} v_{i,k}$ be the maximum value among buyers in group $i$. Then, the following results hold for the optimal allocation:
\begin{enumerate}[label=(\roman*)]
\item If the item is allocated to group $i$, then it is allocated to the buyer in group $i$ with the highest value, i.e., if $x_{i,k} = 1$, then $v_{i,k} = v_i$.
\item The allocation decision at the group level depends only on the maximum value of the two groups, $v_1$ and $v_2$. For any $i$, let $G_i$ denote the pairs $(v_1, v_2)$ for which the item is allocated to group $i$. Then,
there exists $\eta_1, \eta_2 \geq 0$ and $\gamma$ with $\eta_2 = \eta_1 + \gamma$ such that (up to a measure-zero set)
\begin{subequations} \label{eqn:fair_allocation_rule_static}
\begin{align}
G_1 &= \left \{ (v_1, v_2) \mid \phi_1(v_1) \geq \phi_2(v_2) + \gamma \text{ and }  \phi_1(v_1) \geq -\eta_1 \right \},  \\
G_2 &= \left \{ (v_1, v_2) \mid \phi_2(v_2) \geq \phi_1(v_1) - \gamma \text{ and } \phi_2(v_2) \geq -\eta_2 \right \}.
\end{align}
\end{subequations}
\end{enumerate}
\end{theorem}
An example of the optimal allocation rule with $\gamma > 0$ is depicted in Figure \ref{fig:fair_allocation_static} where group one is allocated the item in the yellow and orange regions and group two is allocated the item in the blue and green regions. 

The first part of the theorem is intuitive. The second part follows from the fact that the loss in the seller's utility increases in Euclidean distance from the boundary of the unconstrained optimal allocation $\phi_1(v_1) = \phi_2(v_2)$. Therefore, virtual valuations of the under-allocated group that are closest to those of the over-allocated group, specifically within a distance $\gamma$, are the set that minimize the cost of reallocating to the target group to achieve the desired balance between groups. In the case of insufficient unconstrained allocation, the further modification of allocating in cases where the item remains unallocated in the unconstrained mechanism is necessary. Again, the loss-minimizing set that is reallocated to achieve overall sufficiency in allocation is the set that is within a distance $\eta_1$ from the $\gamma$-shifted boundary in the third quadrant. The optimal mechanism should inefficiently reallocate only as much as necessary, so $\eta_1, \eta_2$, and $\gamma$ are chosen so that the fairness constraint binds. A full proof of this result can be found in the appendix.

As \cref{thrm:fair_allocation_static} shows, the allocation that maximizes expected seller utility while satisfying the fairness constraint \eqref{eqn:static_fairness_constraint} may not allocate the good to the highest bidder. In particular, the optimal fair mechanism subsidizes one group that would otherwise not meet target allocation levels. To maximize revenue while satisfying a given fairness constraint, we deviate from the unconstrained allocation in the most cost efficient way, that is, we reassign the good to a lower valuation bidder in cases that least impact expected seller revenue.

Also, notice that, when there is no fairness constraint, the seller does not allocate the item if the maximum bid is below their \emph{reserve price} $r = \min\{\phi_1^{-1}(0), \phi_2^{-1}(0)\}$. Therefore, it is possible that, in the unconstrained optimal mechanism, the total probability of allocation across both groups be less than the sum of the target allocation levels $\alpha_i$. This is the intuition behind why $\eta_1, \eta_2$ may both be greater than zero in some cases. In other words, we must simultaneously enforce that the item is allocated more often in general and that the relative rates of allocation between groups is sufficiently balanced.

We conclude this section with a result which helps us to identify $\gamma$, $\eta_1$, and $\eta_2$.
\begin{proposition} \label{proposition:optimal_static}
Let $\gamma$, $\eta_1$, and $\eta_2$ correspond to the optimal allocation given by Theorem \ref{thrm:fair_allocation_static}. Then,
\begin{enumerate}
\item If $\gamma > 0$, i.e., the seller strictly subsidizes in favor of the second group, we have $\mathbb{P}(G_2) = \alpha_2$. Moreover, if $\eta_1 >0$, we have $\mathbb{P}(G_1) = \alpha_1$.
\item On the other hand, if $\gamma < 0$, i.e., the seller strictly subsidizes in favor of the first group, we have $\mathbb{P}(G_1) = \alpha_1$. Moreover, if $\eta_2 >0$, we have $\mathbb{P}(G_2) = \alpha_2$.
\end{enumerate}  
\end{proposition}
It is worth noting that the probability $\mathbb{P}(\cdot)$ above is defined with respect to the distribution of the pair of maximum values $(v_1, v_2)$ and its CDF is given by 
\begin{equation} \label{eqn:CDF_max}
F_\text{max}(v_1,v_2) = F_1(v_1)^n F_2(v_2)^n.    
\end{equation}
The intuition behind this proposition is that the seller satisfies the fairness constraint while maximizing their utility by deviating from the optimal unconstrained only as much as is necessary to satisfy the fairness constraint with equality. A full proof can be found in the appendix.
\section{The Dynamic Case}\label{sec:dynamic}
We next turn our attention to the dynamic case. Here, our main approach involves using backward induction. The analysis from the previous section serves as the basis for the optimal allocation at the last round, i.e., $t=T$. In this section, we describe how we proceed backward to find the optimal dynamic auction under the fairness constraint. We make the following assumption which simplifies the characterization of the optimal mechanism and allows us to focus on the insights. In Appendix \ref{sec:relax_assumption_allocation}, and after stating the proof of the optimal allocation result in the dynamic case, we illustrate how our analysis extends to cases where this assumption does not hold.
\begin{assumption} \label{assumption:no_reserve_price}
We assume the seller must allocate the item in any round, except possibly the last one.
\end{assumption}
We start our analysis by making a number of definitions and observations. Suppose we are at the beginning of round $t$. Given \eqref{eqn:fairness_constraint}, our allocation at round $t$ and the future rounds should satisfy the following constraint:
\begin{equation} \label{eqn:residual_fairness_constraint}
\sum_{\tau=t}^{T-1} \delta^{\tau-t} \sum_{k=1}^n x_{i,k}^\tau + \delta^{T-t}  \mathbb{E} \left [ \sum_{k=1}^n x_{i,k}^T \right ] \geq 
R_i^t :=  \frac{1}{\delta^{t-1}} \left ( \frac{1-\delta^T}{1-\delta} \alpha_i - \sum_{\tau=1}^{t-1} \delta^{\tau-1} \sum_{k=1}^n x_{i,k}^\tau \right),   
\end{equation}
We refer to the right-hand side as the \textit{residual minimum allocation} for group $i$. It is noteworthy that equation \eqref{eqn:residual_fairness_constraint} can be interpreted as if the auction is starting at time $t$, with the fairness constraint given by $R_i^t(1-\delta)/(1-\delta^{T-t})$. We next make the following simple observation:
\begin{fact}\label{fact:resid_min_alloc_update}
Given the allocation at time $t$, the residual minimum allocation at time $t+1$ is given by:
\begin{equation*}
R_{i}^{t+1} = 
\begin{cases}
\frac{1}{\delta} (R_i^t - 1) & \text{if the item is allocated to group } i, \\
\frac{1}{\delta}R_i^t & \text{if the item is not allocated to group } i.
\end{cases}
\end{equation*}
\end{fact}
We next introduce two interim functions. For the sake of these two definitions, assume the auction starts at time $t$, i.e., the utility at any time $\tau$, for $\tau \geq t$, is discounted by a factor of $\delta^{\tau-t}$. With this assumption, and for a given pair of residual minimum allocations $(R_1^t, R_2^t)$, we define $\mu^t(R_1^t, R_2^t)$ as the (discounted) sum of the expected utility of the seller at time $t$ and thereafter, under the optimal mechanism satisfying equation \eqref{eqn:residual_fairness_constraint}, prior to the realization of values. Similarly, $\nu_i^t(R_1^t, R_2^t)$ represents the expected utility of a buyer in group $i$. It is worth noting that, given the distribution of values within each group is similar, $\nu_i^t$ does not depend on $k$ (the buyer's index within the group). We also set the these functions equal to $-\infty$ if there is no feasible auction satisfying \eqref{eqn:residual_fairness_constraint}. The following result regarding $\nu_i^T(\cdot, \cdot)$ and $\mu^T(\cdot, \cdot)$ is a corollary of Theorem \ref{thrm:fair_allocation_static}. Notice that the last round can be seen as a static auction with fairness constraint for group $i$ given by $R_i^T$.
\begin{corollary} \label{corollary:last_round}
At round $T$, and for a given pair of residual minimum allocations $(R_1^T, R_2^T)$, we have
\begin{subequations} \label{eqn:mu_nu_T}
\begin{align}
\nu_i^T(R_1^T, R_2^T) &= \frac{1}{n} \int_{G_i} \frac{1-F_i^T(v_i)}{f_i^T(v_i)} d F_\text{max}^T(v_1,v_2), \\
\mu^T(R_1^T, R_2^T) &= \int_{G_1} \phi_1^T(v_1)d F_\text{max}^T(v_1,v_2) +    
\int_{G_2} \phi_2^T(v_2)d F_\text{max}^T(v_1,v_2),
\end{align}    
\end{subequations}
where $F_\text{max}^T(v_1,v_2)$ is given by \eqref{eqn:CDF_max} and $G_i$'s are defined in Theorem \ref{thrm:fair_allocation_static}.
\end{corollary}
This corollary is proved in the appendix. Now, suppose we are in the $t$-th round, conducting the backward induction. Consequently, we have access to $\nu_i^{t+1}(\cdot, \cdot)$ (for $i \in [2]$) and $\mu^{t+1}(\cdot, \cdot)$. Next, we will establish the optimal allocation for round $t$ and describe how to compute $\nu_i^t(\cdot, \cdot)$ and $\mu^t(\cdot, \cdot)$ accordingly. We have three regimes:
\paragraph{(I) When the auction is infeasible:} Suppose that no matter whether the seller gives the item to the first or second group, they would not be able to satisfy the residual minimum allocation constraint in the remaining rounds. In this case, we declare the auction is infeasible at round $t$. 
\begin{fact}
For a given pair of residual minimum allocations $(R_1^t, R_2^t)$, suppose we have 
\begin{equation}
\mu^{t+1}((R_1^t - 1)/\delta, R_2^t/\delta) \text{ and } 
\mu^{t+1}(R_1^t/\delta, (R_2^t-1)/\delta)= 
-\infty.         
\end{equation}
Then, there is no feasible auction satisfying the residual minimum allocation constraints, and hence, we set $\mu^t(R_1^t, R_2^t)$ and $\nu_i^t(R_1^t, R_2^t)$, for both $i \in [2]$, equal to $-\infty$.
\end{fact}
To better highlight this proposition, let us consider a simple example. Suppose $T=2$, $\delta=0.5$, and $\alpha_1=\alpha_2 = 0.4$. In this scenario, $R_1^1 = R_2^1 = 0.6$. Now, the residual minimum allocation for the second round for the group that doesn't receive the item in the first round would be $1.2$, which is not feasible. Hence, there is no feasible auction that satisfies the fairness constraint.
\paragraph{(II) When the item must be allocated to a certain group:} Suppose, for instance, that if the seller allocated the item to the second group, they would not be able to satisfy the resulting residual minimum allocation constraints in the remaining rounds. However, allocating the item to the first group would result in a feasible auction. Therefore, in this case, the seller must allocate the item (if keeping it were feasible, then allocating it to the second group would also have been feasible), and moreover, the item must be allocated to the first group. The following result formalizes the optimal allocation for this scenario.
\begin{proposition}\label{proposition:dynamic_one_group}
Suppose Assumption \ref{assumption:regularity} holds. For a given pair of residual minimum allocations $(R_1^t, R_2^t)$ and some $i \in [2]$, suppose we have
\begin{equation}
\mu^{t+1}(R_i^t/\delta, (R_{-i}^t-1)/\delta) = -\infty
\text{ but } 
\mu^{t+1}((R_i^t - 1)/\delta, R_{-i}^t/\delta) > -\infty .
\end{equation}    
Then, the seller allocates the item to the buyer with the highest value in group $i$, with the payment being equal to the second-highest value within the same group. Furthermore, we have 
\begin{subequations}\label{eqn:dynamic_one_group_updates}
\begin{align}
\nu_{i}^t(R_1^t, R_2^t) &= \int (1-F_{i}^t(v))(F_{i}^t(v))^{n-1} dv  + \delta \nu_{i}^{t+1}((R_{i}^t - 1)/\delta, R_{-i}^t/\delta) , \label{eqn:dynamic_one_group_updates_nu_i}\\
\nu_{-i}^t(R_1^t, R_2^t) &= \delta \nu_{-i}^{t+1}((R_i^t - 1)/\delta, R_{-i}^t/\delta), \\
\mu^t(R_1^t, R_2^t) &= \int \phi_i^t(v) d F_i^t(v)^n + \delta \mu^{t+1}((R_i^t - 1)/\delta, R_{-i}^t/\delta).
\end{align}
\end{subequations}
\end{proposition}
Intuitively, if the item must be allocated to the first group, the question of who within that group receives the item essentially boils down to a static Vickrey auction: the buyer offering the highest value acquires the item and pays the price equal to the second highest value in the group. It is important to note that, in this scenario, there is no reserve price (and we do not need to explicitly impose Assumption \ref{assumption:no_reserve_price}) because the seller is obliged to allocate the item to meet the fairness constraint.
\paragraph{(III) When allocation to both groups is feasible:} Suppose that whether the seller allocates the item to the first or second group, they still have the opportunity to meet the allocation constraint in future rounds. Now, what should be the seller's allocation strategy to maximize their utility? 

Note that the expected utility of buyers in each group changes based on whether or not they receive the item in the current round. The difference between these two scenarios can be interpreted as the expected utility of not receiving the item. 
More formally, for a given pair of residual minimum allocations $(R_1^t, R_2^t)$, and when allocation to both groups is feasible, the net utility of not receiving the item for group $i$ is defined as
\begin{equation} \label{eqn:gain_not_receiving}
\Delta_i^t(R_1^t, R_2^t) := \nu_i^{t+1}(R_i^t/\delta, (R_{-i}^t-1)/\delta)  
- \nu_i^{t+1}((R_i^t - 1)/\delta, R_{-i}^t/\delta).
\end{equation}
\begin{remark}\label{remark:pos_delta}
One might expect the net utility of not receiving the item to be nonnegative for each group, given that not receiving the item at the current time implies a higher chance of receiving it in future rounds. However, interestingly, this argument is not necessarily true; the expected utility of a buyer may decrease even as the fairness constraint on their group increases.
Let us elaborate this matter by a simple example. Consider a scenario with \(T=2\), \(n=1\), and \(\delta=0.8\). Suppose the distributions of values are identical across both groups, i.e., \(F_1^t = F_2^t\) for \(t \in [2]\), while the values from the first round are significantly lower than those from the second round, i.e., \(\bar{v}_i^1 \ll \underline{v}_i^2\). Initially, let \(\alpha_1=\alpha_2=0.2\). In this scenario, each group has some positive probability of receiving the item in the second round. Now, let us increase \(\alpha_1\) to \(0.5\). It is straightforward to verify that the only feasible allocation is to give the item to the first group in the first round and to the second group in the second round. Given that the values in the first round are considerably lower, this allocation actually decreases the expected utility for the first group.
\end{remark}
Lastly, we define $\Delta_0^t(R_1^t, R_2^t)$ as the difference in seller's future utility by allocating the item to the first group, compared to allocating it to the second group, i.e.,
\begin{equation}
\Delta_0^t(R_1^t, R_2^t) := \mu^{t+1}((R_{1}^t-1)/\delta, R_{2}^t/\delta) - \mu^{t+1}((R_{2}^t-1)/\delta, R_{1}^t/\delta). 
\end{equation}
\begin{theorem} \label{theorem:dynamic_both_feasible}
Suppose Assumptions \ref{assumption:regularity} and \ref{assumption:no_reserve_price} hold. For a given pair of residual minimum allocations $(R_1^t, R_2^t)$, suppose both $\mu^{t+1}(R_1^t/\delta, (R_2^t-1)/\delta)$ and $\mu^{t+1}((R_1^t - 1)/\delta, R_2^t/\delta)$ are finite. Then, 
the seller allocates the item to the buyer with the highest value in group $i$ if $(v_1^t, v_2^t) \in G_i^t$, where $v_j^t := \max_{k} v_{j,k}^t$ is the maximum value among buyers in group $j$ and $G_i^t$ is given by 
\begin{align}\label{eqn:G_i_t}
G_i^t := \Big \{ (v_1, v_2) ~\Big \vert~  \phi_i^t(v_i) - \phi_{-i}^t(v_{-i}) \geq n \delta \left ( \Delta_i^t(R_1^t, R_2^t) - \Delta_{-i}^t(R_1^t, R_2^t)\right ) + (-1)^i \delta \Delta_0^t(R_1^t, R_2^t) \Big \}.
\end{align}
Moreover, let $i^*$ denote the index of the winner's group. Then, the payment of buyer $(i,k)$ is given by
\begin{align}
v_{i,k}^t~x_{i,k}^t(v_{i,k}^t, \bv_{-(i,k)}^t) - \int_{\underline{v}_i^t}^{v_{i,k}^t} x_{i,k}^t(z, \bv_{-(i,k)}^t) dz 
+ \delta \Delta_i^t(R_1^t, R_2^t) (\zeta_i^t - \mathbbm{1}(i^*=i)),    
\end{align}
where $\zeta_i^t$ denotes the probability of group $i$ winning the item if the auction runs with $n-1$ buyers from their group instead of $n$ buyers. 
\end{theorem}
The explicit characterization of $\zeta_i^t$'s is provided in Appendix \ref{proof:theorem:dynamic_both_feasible}. Similar to \cref{thrm:fair_allocation_static}, our derivation in Appendix \ref{proof:ironing} shows that we can relax Assumption \ref{assumption:regularity} through ironing. It is worth emphasizing that while the term $\zeta_i^t$ depends on the allocation rule of an auction with $2n-1$ buyers, we can derive it in a non-recursive manner. The reason is that finding the allocation, similar to $G_i^t$ above, does not require running the smaller-size auction. We draw a few insights from this result:
\begin{itemize}[leftmargin=*]
\item Notice that, similar to the static case, the boundary of the allocation rule is a linear function of the maximum virtual values of both groups.
\item Note that the payment function consists of three terms. The first term, $v_{i,k}^t~x_{i,k}^t(v_{i,k}^t, \bv_{-(i,k)}^t) - \int_{\underline{v}_i^t}^{v_{i,k}^t} x_{i,k}^t(z, \bv_{-(i,k)}^t) dz$, represents what only the winner pays and is equal to the minimum bid they could have made to still win the item. This is indeed similar to the second-price auction in which the winner's payment is equal to the second highest bid.

To better understand the other two terms, and for the sake of discussion, suppose $\Delta_i^t(R_1^t, R_2^t)$'s are nonnegative, indicating a nonnegative net future utility associated with not receiving the item at the current time (because it means your chance of receiving it in the future increases). The second term of the payment function is $- \delta \Delta_i^t(R_1^t, R_2^t) \mathbbm{1}(i^*=i)$, which is negative and thus represents a transfer from the seller to the buyers (we call this \textit{participation bonus}). Only the buyers from the group that wins the item receive this payment. To understand why such a reward is needed from the seller, notice that when a buyer's group wins the item, they are in fact losing an amount of $\delta \Delta_i^t(R_1^t, R_2^t)$ in their future utility because their group has a lower chance of winning in the future. Hence, in some sense, this reduces the value of the item in the current round, and the buyers may find it profitable to underreport their value. That said, this payment is there to incentivize truthful reporting.

Finally, the last term in the payment function is $\delta \Delta_i^t(R_1^t, R_2^t) \zeta_i^t$, which is a payment that every buyer should make to the seller. The rationale here is that if the buyer skips the current round, then if their group wins, which happens with probability $\zeta_i^t$, their group would receive a payment equal to $\delta \Delta_i^t(R_1^t, R_2^t)$ from the seller, as we elaborated above. Hence, the seller charges the same amount as the \textit{entry fee}.
\item Notice that in the allocation rule $G_i^t$, there is a threshold of $n \delta \left ( \Delta_i^t(R_1^t, R_2^t) - \Delta_{-i}^t(R_1^t, R_2^t)\right ) + (-1)^i \delta \Delta_0^t(R_1^t, R_2^t)$ which determines which group receives the item. Notice that this threshold increases as $\Delta_i^t(R_1^t, R_2^t)$ increases, meaning that the seller is less willing to give the item to group $i$. This is because a higher $\Delta_i^t(R_1^t, R_2^t)$ means that the seller has to make a higher payment to group $i$'s buyers if they win. Additionally, this threshold decreases for the first group when $\Delta_0^t(R_1^t, R_2^t)$ increases, indicating that the seller is more willing to allocate the item to the first group when $\Delta_0^t(R_1^t, R_2^t)$ is higher. This is intuitive as this term represents the seller's future net utility by allocating the item to the first group compared to allocating it to the second group.
\end{itemize}
We also derive the update of the interim functions under the premise of Theorem \ref{theorem:dynamic_both_feasible}:
\begin{subequations}\label{eqn:dynamic_two_groups_updates}
\begin{align}
& \nu_{i}^t(R_1^t, R_2^t) = \frac{1}{n} \int_{G_i^t} \frac{1-F_i^t(v_i)}{f_i^t(v_i)} d F_\text{max}^t(v_1,v_2)
+\delta \nu_i^{t+1}(R_i^t/\delta, (R_{-i}^t-1)/\delta) - \delta \Delta_i^t(R_1^t, R_2^t) \zeta_i^t, \label{eqn:dynamic_two_groups_updates_nu}\\
& \mu^t(R_1^t, R_2^t) = \\
& \sum_{i=1}^2 \Bigg ( \int_{G_i^t} \phi_i^t(v_i^t)~dF_\text{max}^t(v_1,v_2)
+ \delta \mathbb{P}(G_i^t) \mu^{t+1}((R_{i}^t-1)/\delta, R_{-i}^t/\delta) 
+ \delta n \Delta_i^t(R_1^t, R_2^t) (\zeta_i^t- \mathbb{P}(G_i^t)) \Bigg ), \nonumber
\end{align}
\end{subequations}
where the probability measure $\mathbb{P}(\cdot)$ is taken with respect to the distribution $F_\text{max}^t(v_1,v_2) := F_1^t(v_1)^n F_2^t(v_2)^n$.

\begin{remark}[\textbf{The ex post fairness constraint}] \label{remark:ex_post}
Note that if we were to consider the ex post fairness constraint, the recursive nature of our analysis would remain the same, with the only difference being the basis for the backward induction, i.e., the case of $t=T$, or the static case. In that case, it is not feasible to impose an ex post fairness constraint unless the requirement applies to only one of the groups, i.e., $R_1^T = 0$ or $R_2^T = 0$, in which case we must always allocate the item to the group with the positive minimum allocation requirement. In other words, under the ex post fairness constraint, our recursive derivations \eqref{eqn:dynamic_two_groups_updates} would still hold, but \cref{corollary:last_round} would need to be modified to reflect the solution of the static case under the ex post fairness constraint, as described above. 
\end{remark}

\subsection{Do the entry fees cover the participation bonuses?}  

As we have seen so far, the seller's utility consists of three components:  payment from the winning buyer, entry fees, and participation bonuses. The first of these components matches the classical second price auction setting; The latter two components, however,  extend the classical setting and raise questions about their net effect. In other words, since the optimal fair mechanism includes additional transfers both from the seller to the buyer and vice versa, does the seller end up paying the buyers out of pocket?  One can indeed verify that this could happen by considering the case where $n=1$, where the entry fee is zero but the participation bonus is positive. However, we show that, as the number of buyers increases, the entry fee offsets the participation bonus the seller has to pay.

 
For this analysis, we assume the setting where the expected future utility of a buyer increases if they do not win in the current round. In other words, we suppose that we are not in the type of corner case described in \cref{remark:pos_delta}.
\begin{assumption}\label{assumption:pos_Delta}
For any $i$ and $t$, a buyer in group $i$'s expected future utility over the remaining rounds increases if they do not win the item in round $t$, i.e., 
 \begin{equation*}
     \begin{aligned}
         \Delta_i^t(R_1^t, R_2^t) \geq 0.
     \end{aligned}
 \end{equation*}
 \end{assumption}
 Next we present a result that characterizes the behavior of the expected difference between the entry fees and participation bonuses in terms of the number of buyers participating in the auction.
\begin{proposition}\label{prop:seller_payment}
Suppose all buyers' values are bounded and Assumption \ref{assumption:pos_Delta} holds. Then, the total participation bonuses exceed the total entry fees by an amount of $\mathcal{O}(1/n)$ in expectation, i.e.,  
\begin{equation*}
    \begin{aligned}
        \mathbb{E}\left[ n\delta \sum_{i=1}^2 \Delta_i^t(R_1^t, R_2^t) (\zeta_i^t - \mathbbm{1}(i^*=i))\right] \geq \frac{-C}{n},
    \end{aligned}
\end{equation*} 
for a constant $C$ depending on $\delta $ and the maximal buyer value.
\end{proposition} 
This result is built upon two observations. The first is that the difference between the probability that group $i$ wins in round $t$  with $n-1$ buyers and the probability that group $i$ wins in round $t$ with $n$ buyers is at least $-1/n$, i.e, for all $i$ and $t$, 
\begin{equation*}
    \begin{aligned}
        \zeta_i^t - \mathbb{P}(G_i^t) \geq \frac{-1}{n}
    \end{aligned}
\end{equation*}
The second step is to see that the expected utility of buyer $i$ at round $t$ is bounded by $\mathcal{O}\left(\frac{1}{n}\right)$, i.e., 
\begin{equation*}
    \begin{aligned}
        \nu_i^t(R_1^t, R_2^t) \leq \mathcal{O}\left(\frac{1}{n}\right)
    \end{aligned}
\end{equation*}
For a full proof, see Appendix \ref{proof:seller_payment}. This result shows that, as the size of the auction grows, the entry fee paid to the seller by the buyers covers the participation bonus earned by buyers in the winning group. In other words, the seller pays, at worst, only a negligible amount when there are many buyers. 
\subsection{An Approximation Scheme}
Our results so far offer an exact characterization of the optimal dynamic allocation. However, due to the recursive nature of the analysis, the computational complexity of this allocation increases exponentially with the number of rounds, $T$. We conclude this section by discussing efficient methods to find an approximately optimal allocation. Since our focus is not on approximating the integrals here, we assume we have access to an oracle which can compute integrals and solve integral equations. 
\begin{assumption}\label{assumption:oracle}
We have access to an oracle that, for any $i \in [n]$ and $t \in [T]$, can compute the integrals $\int_G \phi_i^t(v_i^t) , dF_\text{max}^t(v_1, v_2)$ and $\int_G \frac{1 - F_i^t(v_i)}{f_i^t(v_i)} , dF_\text{max}^t(v_1, v_2)$ over any affine subspace $G$. Furthermore, the oracle can find the optimal allocation in the static case by solving the integral equations in Proposition \ref{proposition:optimal_static}.
\end{assumption}
We start by making the following observation regarding the case $\delta=1$.
\begin{fact}\label{fact:delta_equal_1_complexity}
Suppose $\delta=1$ and Assumptions \ref{assumption:regularity}, \ref{assumption:no_reserve_price} , and \ref{assumption:oracle} hold.  Then, we can find the exact optimal allocation by calling the oracle $\mathcal{O}(T^2)$ times.     
\end{fact}
To see why this holds, notice that, for the case of $\delta=1$, the pair $(R_1^t, R_2^t)$ at most takes $T$ different values for any $t$, and so computing all of $\{\mu^t(R_1^t, R_2^t), \nu_1^t(R_1^t, R_2^t), \nu_2^t(R_1^t, R_2^t)\}_t$ would  require $\mathcal{O}(T^2)$ rounds of computation using the recursive derivations we stated earlier. However, for the case $\delta<1$, this is no longer the case, as the pair $(R_1^t, R_2^t)$ can take on as many as $2^T$ different values. In this case, we can technically choose some $T_0 < T$ and aim to achieve the fairness constraint approximately in the first $T_0$ rounds, and then run the regular second price auction in the remaining rounds. The following result formalizes this.
\begin{proposition}\label{proposition:approx_delta_leq1}
Suppose Assumptions \ref{assumption:regularity}, \ref{assumption:no_reserve_price} , and \ref{assumption:oracle} hold and that $\delta < 1$. Then, for any $\varepsilon \in [0, \min_i \alpha_i]$, there exists an allocation $\bx'$ with the following properties: 
\begin{enumerate}
    \item $\bx'$ guarantees that the fairness constraint for group $i$ is satisfied at a level of at least $(1-\varepsilon)(\alpha_i - \varepsilon)$.
    \item $\bx'$ guarantees that the seller's total utility is at least that of the optimal allocation.
    \item $\bx'$ can be computed by calling the oracle $\mathcal{O}\left( \frac{1}{\varepsilon}^{\frac{1}{\log(1/\delta)}}\right)$ times.
\end{enumerate}
\end{proposition}
Roughly speaking, our approximation works in the following way: We first assume that the auction runs only for $T_0 = \log(\varepsilon)/\log(\delta)$ rounds. We use the recursive interim functions defined earlier to find the exact fair allocation under the ex post fairness constraint at level $\alpha_i - \varepsilon$ for group $i$ over $T_0$ rounds. This would require $\mathcal{O}(2^{T_0})$ calls of the oracle. Furthermore, we establish that the slight reduction in the fairness constraint from $\alpha_i$ to $\alpha_i - \varepsilon$, along with the choice of $T_0$, ensures that the allocation we find delivers a seller utility at least at the level of the global optimal allocation. Lastly, for the remaining $T - T_0$ rounds, we conduct the standard second-price auction. A complete proof can be found in \cref{proof:approximation1}. 

Notice that the complexity of the approximation given by \cref{proposition:approx_delta_leq1} grows large for $\delta$ close to 1. Next, we present a second approximation scheme for the case where $\delta = 1 - \delta'$ and $\delta' \simeq 0$.
\begin{proposition}\label{proposition:approx_delta_approx1}
Suppose Assumptions \ref{assumption:regularity}, \ref{assumption:no_reserve_price} , and \ref{assumption:oracle} hold and that $\delta < 1$. Then, for a fixed constant $c < \delta^2$, there exists an allocation $\bx''$ with the following properties: 
\begin{enumerate}
    \item $\bx''$ satisfies the fairness constraint for group $i$ at a level of at least $c \cdot \alpha_i$.
    \item $\bx''$ guarantees that the seller's total utility is at least $c$ times that of that of the optimal allocation.
    \item $\bx''$ can be computed by calling the oracle $\text{Poly}\left( \frac{1}{1 - \delta} \right)$ times where the polynomial's degree and coefficients depend on $c$.
\end{enumerate}
\end{proposition}
This $\bx''$ modifies the approximation $\bx'$ presented in \cref{proposition:approx_delta_leq1} by introducing a discontinuous discounting technique. In $\bx''$, we partition the first $T_0$ rounds into buckets that use a common approximate discount factor such that the computational complexity of each bucket may be controlled in the same manner as  Fact \ref{fact:delta_equal_1_complexity}. A complete proof statement that relates the early stopping approximation with the discounting approximation that can be tuned to achieve an overall approximation level $c$ can be found in \cref{proof:approximation2}.
\section{Extension to Multiple Groups}%

We have primarily been working in a setting with two groups of buyers. In this section, we demonstrate how our insight regarding achieving fairness through subsidizing in the space of virtual values extends to the general case with $L$ groups of buyers. In particular, the following result generalizes our findings from the static case, as presented in \cref{thrm:fair_allocation_static}, to the multi-group setting. We will maintain all previously established notations, with the difference that the group index $i$ now belongs to the set $[L]$ rather than $[2]$. 
\begin{theorem}\label{theorem:multi_group_static}
Suppose that Assumption \ref{assumption:regularity} holds. For any $i \in [L],$ let $v_i \coloneqq \max_{k} v_{i,k}$ be the maximum value among buyers in group $i$. Then, the following results hold for the optimal allocation:
\begin{enumerate}[label=(\roman*)]
    \item If the item is allocated to group $i$, then it is allocated to the buyer in group $i$ with the highest value, i.e., if $x_{i,k} = 1$, then $v_{i,k} = v_i$.
    \item The allocation decision at the group level depends only on the maximum value of each of the $L$ groups, $\{ v_i\}_{i=1}^L$. For any $i$, let $G_i$ denote the $L$-tuples $(v_1, \hdots, v_L)$ for which the item is allocated to group $i$. Then,
    there exist nonnegative values $\eta_1, \hdots, \eta_L$ 
    such that (up to a measure-zero set) for $i \in [L]$ ,
    \begin{equation*}
        \label{eqn:fair_allocation_rule_static}
    \begin{aligned}
    G_i &= \left \{ (v_1, \hdots, v_L) \mid \phi_i(v_i) + \eta_i \geq 0 \text{ and }  \phi_i(v_i) + \eta_i \geq \phi_j(v_j) + \eta_j \text{ for } i \neq j  \right \}.
    \end{aligned}
    \end{equation*} 
    \end{enumerate}
\end{theorem}
Notice that the insight behind this result is very similar to that of \cref{thrm:fair_allocation_static} for two groups. Here, $\eta_i$ represents the subsidy to group $i$, and $\min_i \eta_i$ can be seen as the overall subsidy to society (similar to the $\gamma$ term in \cref{thrm:fair_allocation_static}).

For the dynamic case, the optimal fair mechanism once again is attained through backward induction. Observe that the result of \cref{corollary:last_round} generalizes to the $L$ group case with the modifications that the expected utilities at round $T$ are a function of the $L$-tuple of residual minimum allocations, $\mathbf{R}^T\coloneqq (R_1^T, \hdots R_L^T),$ and $F_\text{max}^T(v_1,\hdots, v_L) \coloneqq \prod_{i=1}^L F_i(v_i)^n$. The case where the auction is not feasible no matter the seller's allocation and the case where the auction is only feasible if the seller allocates to a particular group match those of the dynamic two group results. Next, we consider the case where the auction remains feasible if the seller allocates to one of multiple groups. 
First, let $\Delta^t_i(j, k)$ denote the loss in a buyer in group $i$'s expected utility in rounds $t+1$ and thereafter when the seller allocates to group $k$ instead of group $j$ in round $t$, and i.e., 
\begin{equation*}
    \begin{aligned}
        \Delta^t_i(j, k) \coloneqq \nu_{i}^{t+1}\left( (R_{j}^t-1)/\delta, R_{-j}^t/\delta \right) - \nu_{i}^{t+1}\left( (R_{k}^t-1)/\delta, R_{-k}^t/\delta \right)
    \end{aligned}
\end{equation*}
Similarly, let $\Delta_0(i, j)$ denote the loss in the seller's expected utility in rounds $t+1$ and thereafter when the seller allocates to group $j$ instead of group $i$ in round $t$, i.e., 
\begin{equation*}
    \begin{aligned}
        \Delta^t_0(i, j) \coloneqq \mu^{t+1}\left( (R_{i}^t-1)/\delta, R_{-i}^t/\delta \right) - \mu^{t+1}\left( (R_{j}^t-1)/\delta, R_{-j}^t/\delta \right)
    \end{aligned}
\end{equation*}
Now, the following result characterizes the allocation rule of the optimal mechanism in the case where the auction remains feasible if the seller allocates to any group. Note that this result can be easily modified for the case where it is feasible to allocate to any subset of at least two groups. 

\begin{theorem}\label{theorem:dynamic_multi}
Suppose Assumptions \ref{assumption:regularity} and \ref{assumption:no_reserve_price} hold. For a given $L$-tuple of residual minimum allocations, $\mathbf{R}^T\coloneqq (R_1^T, \hdots R_L^T)$, suppose suppose that, for all $i \in [L],$ $\mu^{t+1} \left(\left(R_{i}^t-1\right)/\delta, R_{-i}^t/\delta\right)$ is finite where $\left(\left(R_{i}^t-1\right)/\delta, R_{-i}^t/\delta\right)$ denotes the $L$-tuple updated according to Fact \ref{fact:resid_min_alloc_update}. Then, 
the seller allocates the item to the buyer with the highest value in group $i$ if $(v_1^t, \hdots, v_L^t) \in G_i^t$, where $v_j^t := \max_{k} v_{j,k}^t$ is the maximum value among buyers in group $j$ and $G_i^t$ is given by 
\begin{align}\label{eqn:G_i_t_multi}
G_i^t := \Big \{ (v_1^t, \hdots, v_L^t) ~\Big \vert~   \phi_i^t(v_i) - \phi_{j}^t(v_{j}) \geq  \delta n \sum_{\ell = 1}^L \Delta_{\ell}^t(j, i) + \delta \Delta_0^{t}(j, i) \: \forall j \in [L], j \neq i \Big \}.
\end{align}
\end{theorem}
We see that, once again, the optimal fair allocation among multiple groups involves subsidizing otherwise under-allocated groups. A full proof of this result is given in the appendix. 
\section{A Numerical Experiment}
Here, we present the results of a numerical experiment assessing the impact on utilities of varying the fairness constraints of each group. In particular, we implement the case where $\delta = 0.99$, $n = 1$, $v_1^t \sim \text{Uniform}(0.5, 1.5)$, and $v_2^t \sim \text {Uniform}(0, 1)$ for $t \in [2]$ for $T=2, 3, 4$. 

For each value of $T$, we consider combinations of fairness constraints on a discretized grid where $\alpha_1$ and $\alpha_2$ range over $(0, 0.5)$ with increments of $0.1$. For each pair $\alpha_1, \alpha_2$, we calculate the mean difference in utility between the optimal fair allocation at level $\alpha_1, \alpha_2$ and the unconstrained optimal allocation satisfying Assumption \ref{assumption:no_reserve_price} (i.e., $\alpha_1, \alpha_2 = 0$), for the seller and buyers over 10,000 iterations of the mechanism. Note that, for these distributions, when $T=2$ the average unconstrained allocation probabilities are $0.69$ and $0.31$ for groups one and two, respectively. The results for $T=2$ are reported in Figures \ref{fig:experiment_T2}-\ref{fig:experiment_T4}.

We see that the seller utility is decreasing in $\alpha_1, \alpha_2$ but only after the point at which the constraints bind. For group one, we see that, for a fixed $\alpha_1$, utility is decreasing in $\alpha_2$, since, to satisfy the fairness constraint, we re-allocate higher value regions to group two and allocate lower-value regions that would otherwise go unallocated to group one. In contrast, for group two, we see that, for a fixed $\alpha_2$, utility is relatively constant in $\alpha_1$ since the optimal allocation tends to reallocate to group one from the lower-value no-allocation region. 

\begin{figure}[t]
\centering
\begin{subfigure}{.33\textwidth}
  \centering
  \includegraphics[width=\linewidth]{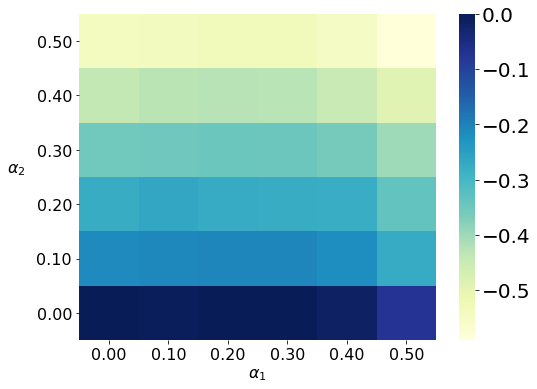}
  \caption{Seller utility}
  \label{fig:experiment_seller}
\end{subfigure}%
\begin{subfigure}{.33\textwidth}
  \centering
  \includegraphics[width=\linewidth]{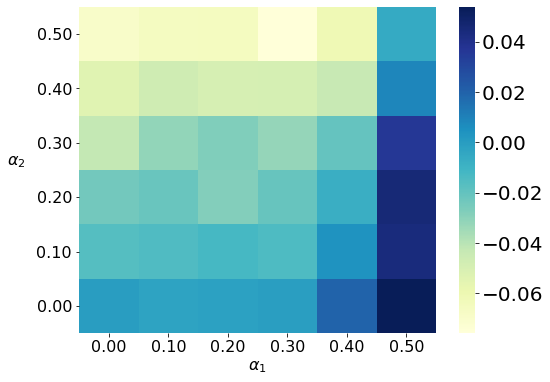}
  \caption{Group One utility}
  \label{fig:experiment_group1}
\end{subfigure}
\begin{subfigure}{.33\textwidth}
  \centering
  \includegraphics[width=\linewidth]{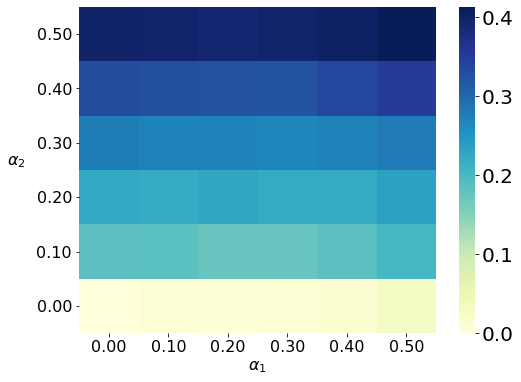}
  \caption{Group Two utility}
  \label{fig:experiment_group2}
\end{subfigure}
\caption{ Difference in utility relative to unconstrained optimal allocation for $T=2$}
\label{fig:experiment_T2}
\end{figure}
\begin{figure}[t]
\centering
\begin{subfigure}{.33\textwidth}
  \centering
  \includegraphics[width=\linewidth]{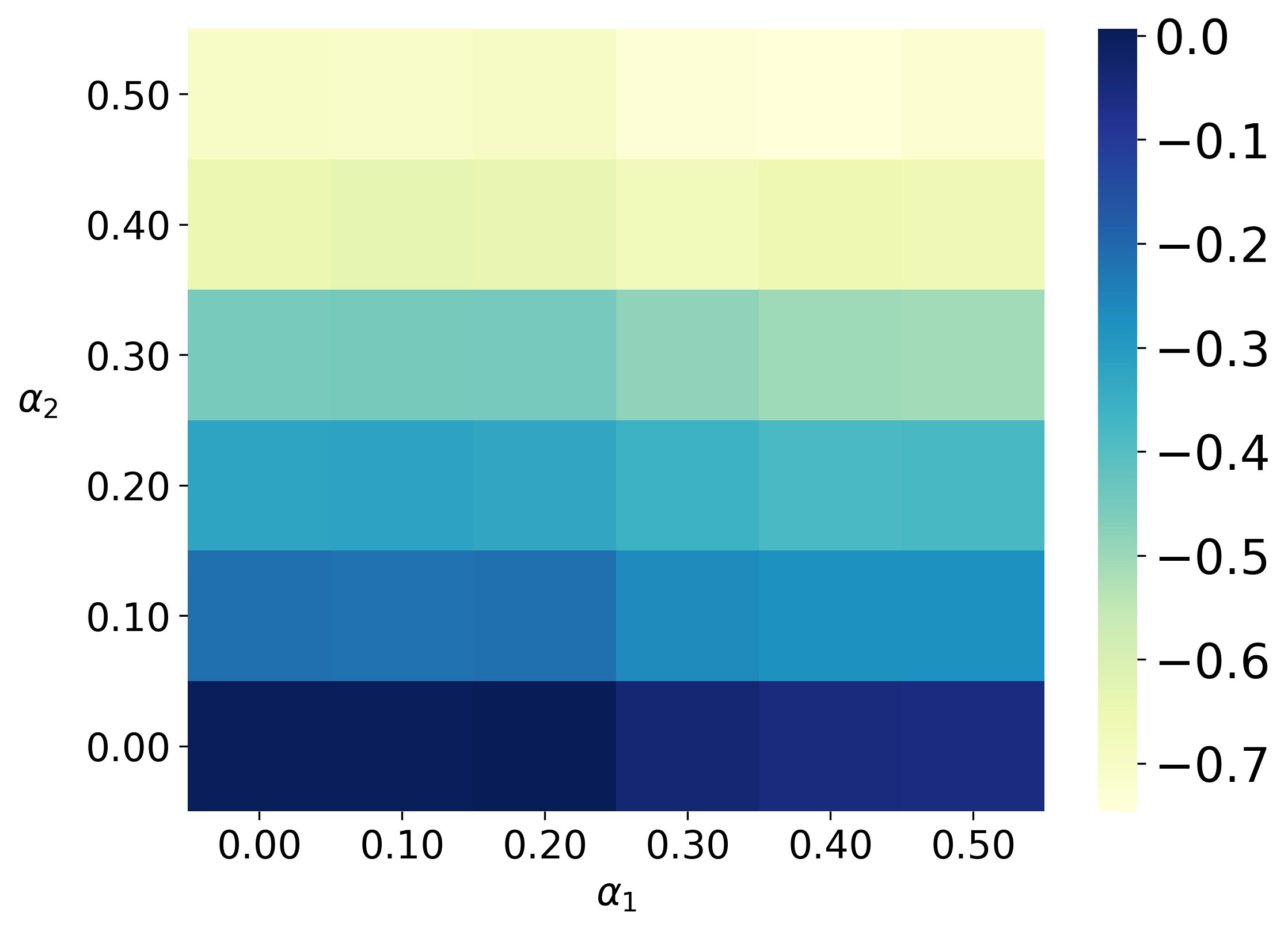}
  \caption{Seller utility}
  \label{fig:experiment_seller}
\end{subfigure}%
\begin{subfigure}{.33\textwidth}
  \centering
  \includegraphics[width=\linewidth]{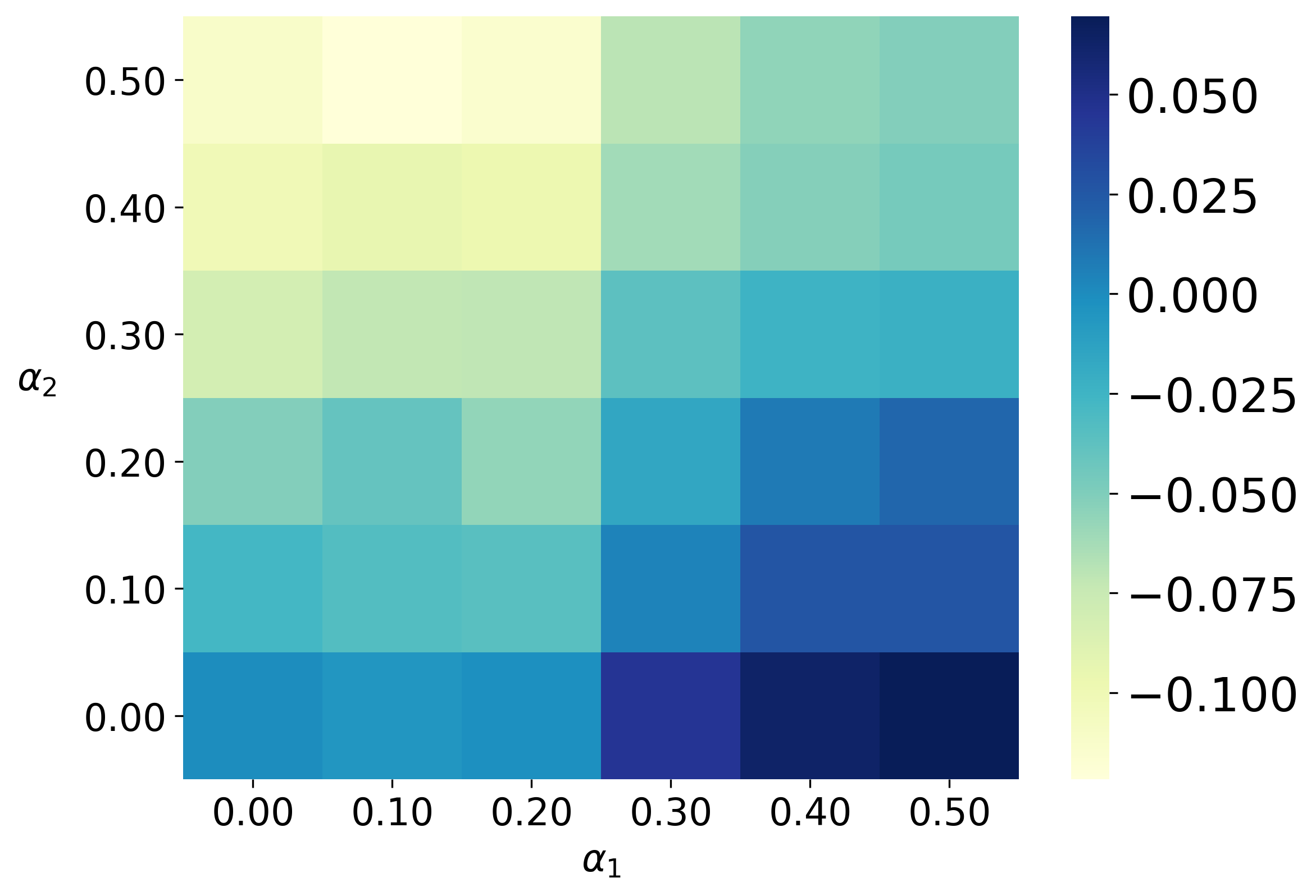}
  \caption{Group One utility}
  \label{fig:experiment_group1}
\end{subfigure}
\begin{subfigure}{.33\textwidth}
  \centering
  \includegraphics[width=\linewidth]{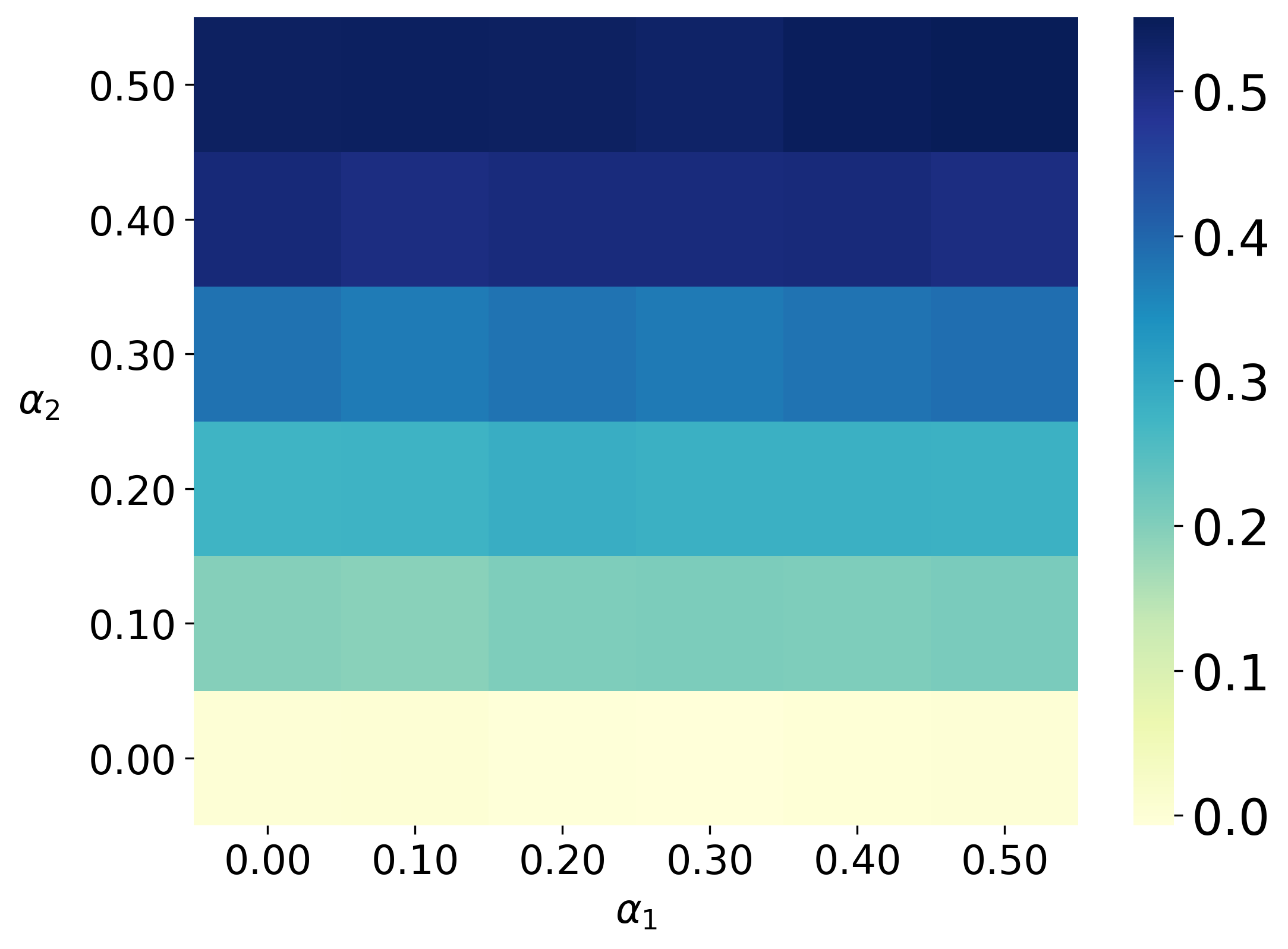}
  \caption{Group Two utility}
  \label{fig:experiment_group2}
\end{subfigure}
\caption{ Difference in utility relative to unconstrained optimal allocation for $T=3$}
\label{fig:experiment_T3}
\end{figure}
\begin{figure}[t]
\centering
\begin{subfigure}{.33\textwidth}
  \centering
  \includegraphics[width=\linewidth]{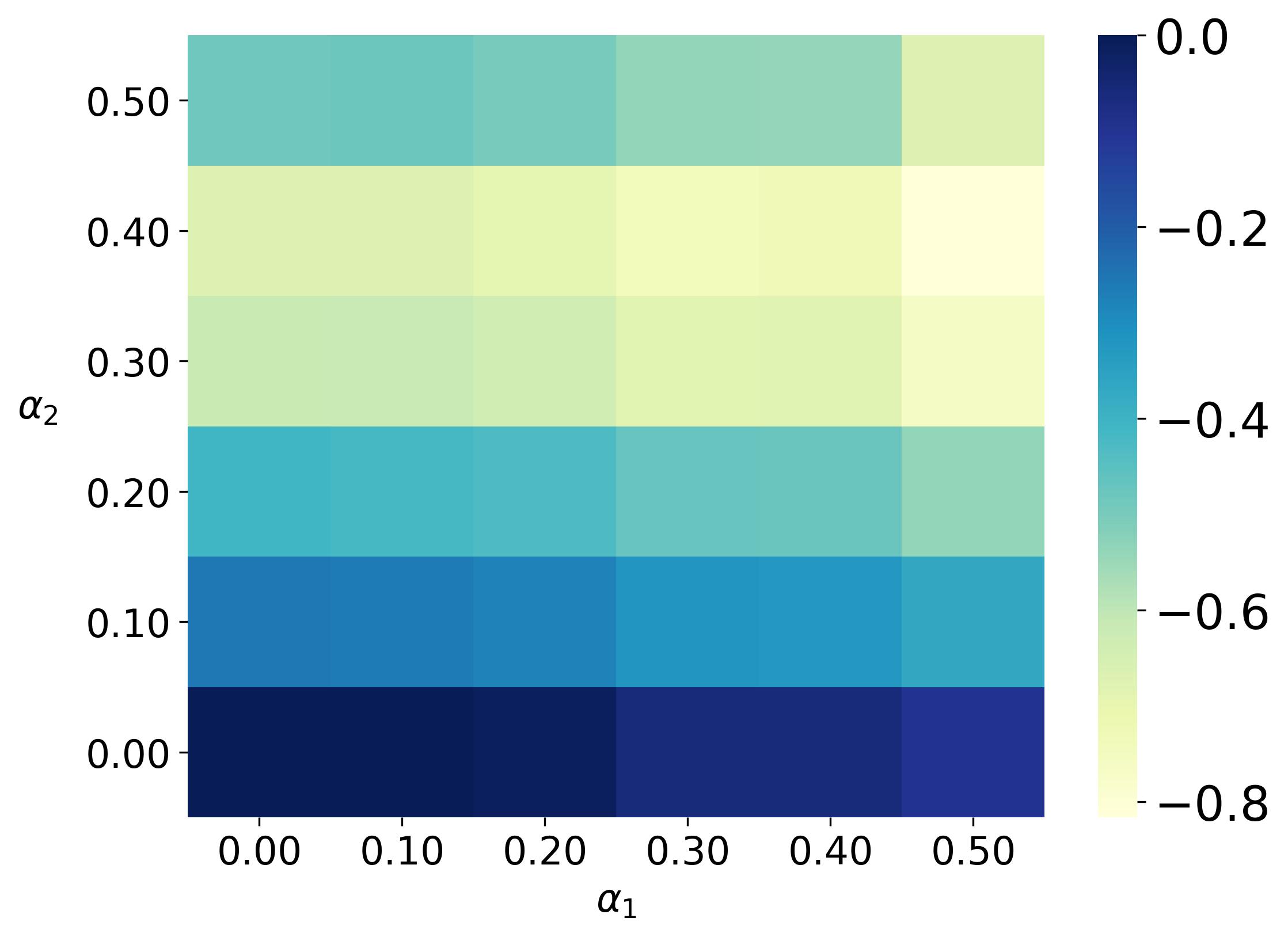}
  \caption{Seller utility}
  \label{fig:experiment_seller}
\end{subfigure}%
\begin{subfigure}{.33\textwidth}
  \centering
  \includegraphics[width=\linewidth]{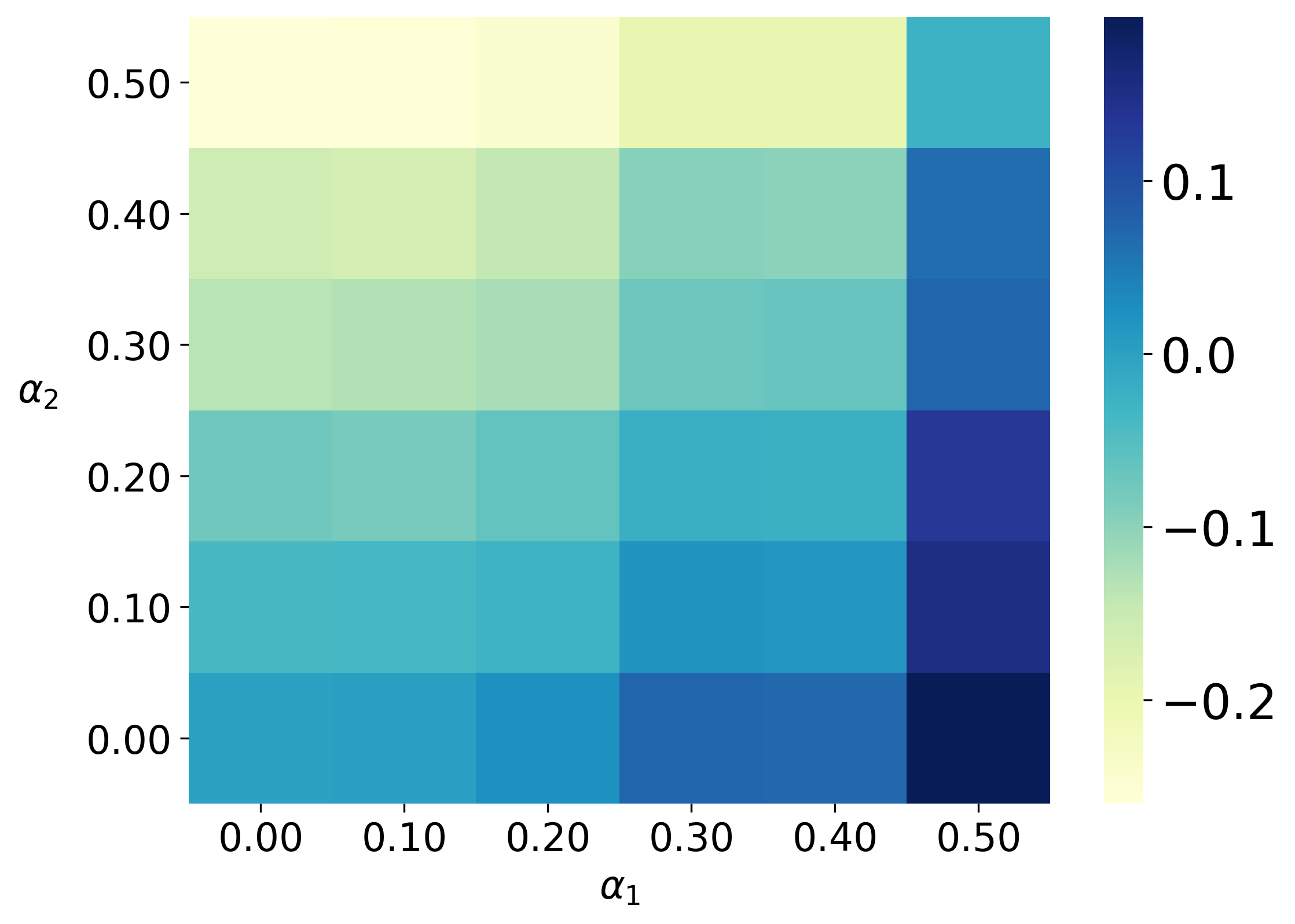}
  \caption{Group One utility}
  \label{fig:experiment_group1}
\end{subfigure}
\begin{subfigure}{.33\textwidth}
  \centering
  \includegraphics[width=\linewidth]{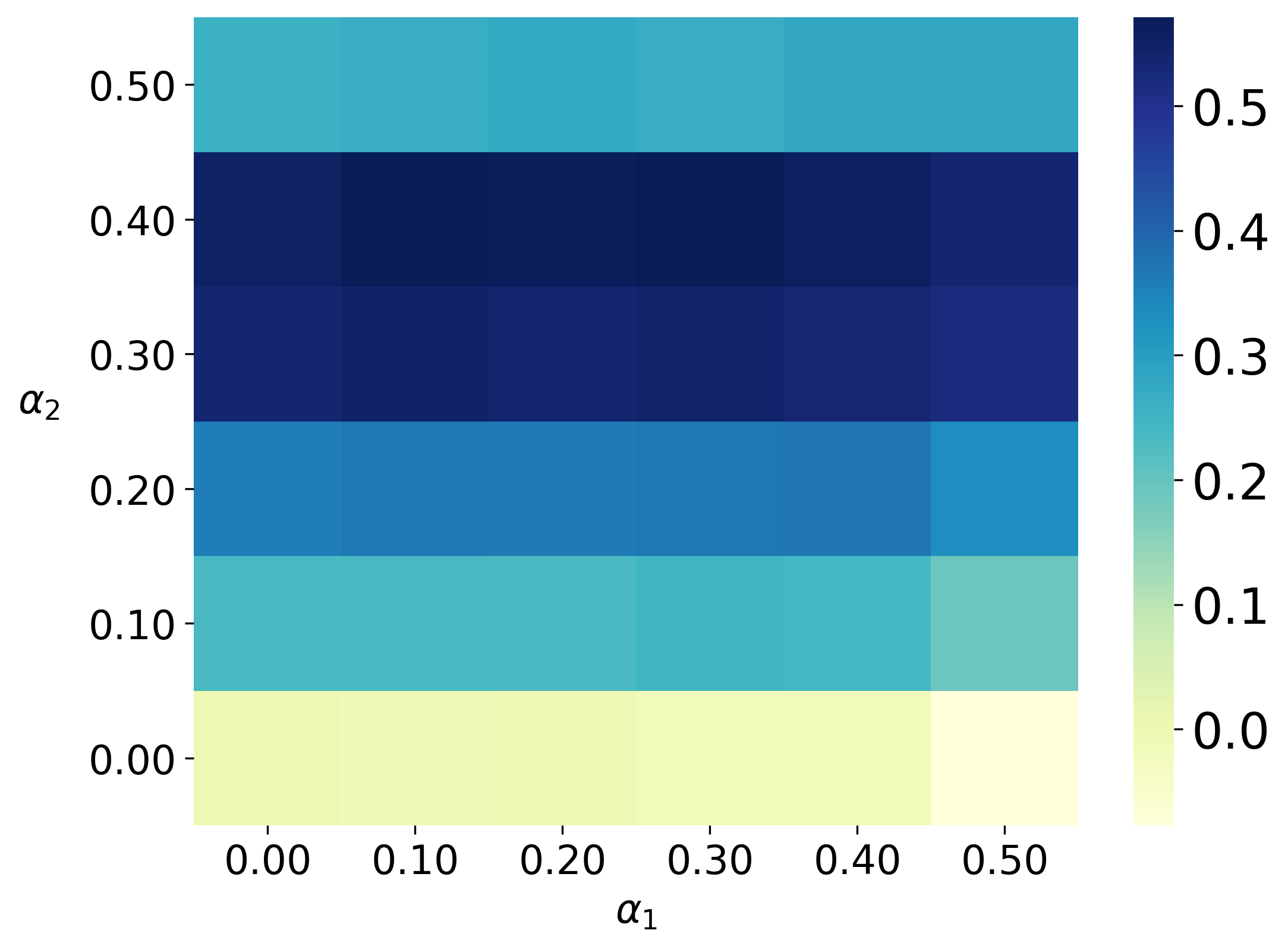}
  \caption{Group Two utility}
  \label{fig:experiment_group2}
\end{subfigure}
\caption{ Difference in utility relative to unconstrained optimal allocation for $T=4$}
\label{fig:experiment_T4}
\end{figure}

\section{Conclusion and Future Directions}
Our paper has explored the integration of group fairness into a dynamic auction design setting. We characterized the revenue-maximizing allocation rule, showing it involves subsidization in favor of groups that would otherwise not be allocated the item enough. We further established that the payment function rewards buyers for participation when their group wins the item, while charging them an entry fee regardless of the allocation and that, as the number of buyers grows, the entry fee covers the participation bonus such that the seller does not pay out of pocket. Moreover, to address the computational complexities of the dynamic setting, we proposed an approximation scheme capable of achieving near-optimal fairness efficiently. 
An important direction for extensions of our work is studying fair dynamic allocations when buyer values are interdependent between buyers, or between rounds. Another area of interest for future work is to consider the \emph{price of fairness}, or the impact to buyer and seller utilities under the fairness constraint.

\section{Acknowledgement}
The authors thank Scott Kominers and Rakesh Vohra for insightful discussion and comments. Alireza Fallah acknowledges support from the European Research Council Synergy Program, the National Science Foundation under grant number DMS-1928930, and the Alfred P. Sloan Foundation under grant G-2021-16778. The latter two grants correspond to his residency at the Simons Laufer Mathematical Sciences Institute (formerly known as MSRI) in Berkeley, California, during the Fall 2023 semester. Michael Jordan acknowledges support from the Mathematical Data Science program of the Office of Naval Research under grant number N00014-21-1-2840 and the European Research Council (ERC-2022-SYG-OCEAN-101071601). Annie Ulichney's work is supported by the National Science Foundation Graduate Research Fellowship Program under Grant No. DGE 2146752. Any opinions, findings, and conclusions or recommendations expressed in this material are those of the author(s) and do not necessarily reflect the views of the National Science Foundation or the European Research Council.

\bibliography{References}
\newpage
\appendix

\section{Extended Related Work}\label{sec:related_work}

\paragraph{Group Fairness} Where much of the classical fair division work focuses on individual fairness, our work builds upon those that extend fair allocation methods to the setting of group fairness for pre-existing groups. {Note that our approach to group-wise fairness aligns with that of group fairness in machine learning which enforces fairness across pre-existing groups. This is in contrast to the approach of a related literature on group-wise fair discrete division, which enforces fairness across all possible groups, as in, e.g., \citet{conitzer2019group} and \citet{aziz2019almost}} Related works addressing group allocative fairness consider notions of group fairness that are fundamentally different from our proportional sense, including group envy-freeness \citep[]{manurangsi2017asymptotic, kyropoulou2020almost, benabbou2019fairness} and democratic group fairness \citep{segal2019democratic}.

\paragraph{Fairness and Machine Learning} More broadly, our work joins a burgeoning literature on fair mechanism design with an algorithmic focus which draws from notions of fairness studied from computer science and machine learning \citep{finocchiaro2021bridging}. The work of \cite{kuo2020proportionnet} applies deep learning techniques to enforce fairness in auctions efficiently, and \cite{deng2022fairness} consider fairness amidst algorithmic bidding. Several works in this segment of the literature address the concern of fairness in online advertising auctions ~\citep[see, e.g.,][]{chawla2020individual, celis2019toward, ilvento2020multi, birmpas2021fair, hadiji2020diversity}.

\paragraph{Fair Matching, Ranking, and Assortment} The setting that motivates the consideration of fair resource allocation is similar to the setting addressed by works on fairness in matching markets, fair ranking, and fair assortment planning \citep{ma2023fairness1,chen2022fair, lu2023simple}. 

\paragraph{Constrained Buyers} Another group of related works studies the impact of budgetary and liquidity constraints on auction revenue~\citep[see, e.g.,][]{malakhov2008optimal, dobzinski2012multi}. Their motivation for considering budget constraints aligns with the setting where fair allocation guarantees are relevant. The setting of asymmetric budgetary constraints studied by \citet{pai2014optimal} motivates the relevance of our work's introduction of guarantees for relative allocative fairness between groups. We build upon this set of work by introducing a framework where the allocative limitations of both homogeneous and heterogeneous budget constraints can be simultaneously overcome as efficiently as possible.

\section{Proofs}\label{sec:proofs} 
In this appendix, we present proofs that are omitted from the body of the paper. 

\subsection{Proof of Theorem \ref{thrm:fair_allocation_static}}\label{proof:static_alloc}
Let $\mathcal{S}$ be the set of $(\phi_{i}(v_{i,k}))_{i,k}$ for which the seller allocates the item to one of the groups in the optimal fair allocation $\tilde{x}$, i.e.,
\begin{equation}
\mathcal{S} = \left \{ (\phi_{i}(v_{i,k}))_{i,k} \Big \vert
\exists i,k \text{ s.t. } \tilde{x}_{i,k}(\bv) = 1
\right \}.
\end{equation}
Recall that $v_i = \max_{k} v_{i,k}$.
We first establish that, given $\mathcal{S}$, the boundary of the optimal allocation is in the form of $\phi_1(v_1) = \phi_2(v_2) + \gamma$ for some $\gamma$. Note that the set $\mathcal{S}$ is measurable as the allocation function is measurable.  

Notice that the probability of $\mathcal{S}$ is at least $\alpha_1 + \alpha_2$. Since $f_1(\cdot)$ and $f_2(\cdot)$ are continuous, we can find an allocation rule that is confined to $\mathcal{S}$, satisfies the fairness constraint, and its boundary takes the form $\phi_1(v_1) = \phi_2(v_2) + \gamma$ for some $\gamma$ (we call such allocations as \textit{affine allocations}). That is, group one receives the object if and only if $\phi_1(v_1) \geq \phi_2(v_2) + \gamma$ and $(\phi_1(v_1), \phi_2(v_2)) \in \mathcal{S}$; similarly, group two receives the object if and only if $\phi_1(v_1) < \phi_2(v_2) + \gamma$ and $(\phi_1(v_1), \phi_2(v_2)) \in \mathcal{S}$. We further assume this allocation allocates the item to the buyer with highest value within each group. Denote this allocation by $x^*$. If there are multiple allocation rules of this form, we select the one whose corresponding $\gamma$ has the smallest absolute value. One can verify by inspection that such allocation is monotone, and hence, satisfies the EPIC condition with its corresponding payment identity. Without loss of generality, we may assume $\gamma > 0$, as the case for $\gamma < 0$ can be argued similarly.

If we had to allocate every pair in $\mathcal{S}$ but there were no fairness constraints, the optimal allocation would have been by checking to the buyer with the highest value, which from groups' point of view would mean the boundary $\phi_1(v_1) \lesseqgtr \phi_2(v_2)$. 
Let $A^*$ be the region that $x^*$ allocates differently from this unconstrained allocation over $\mathcal{S}$, i.e., 
\begin{equation} \label{eqn:A^star}
A^* = \left \{(\phi_{i}(v_{i,k}))_{i,k} \in \mathcal{S} \Big \vert \phi_1(v_1) \geq \phi_2(v_2) \geq \phi_1(v_1) - \gamma \right \}.
\end{equation}
We define $\tilde{A}$ similarly for the allocation $\tilde{x}$ as the set of values for which we allocate the item suboptimally. Also, for any $j \in [2]$, let $Q_j \subset \mathcal{S}$ be the region that is allocated to group $j$ in the optimal unconstrained allocation, i.e., 
\begin{equation}
Q_j = \left \{(\phi_{i}(v_{i,k}))_{i,k} \in \mathcal{S} \Big \vert \phi_j(v_j) \geq \phi_{-j}(v_{-j}) \right \}.    
\end{equation}
In particular, note that $A^* \subseteq Q_1$. We next make the following claim.
\begin{claim} \label{claim:static_proof_1}
The probability of $\tilde{A} \cap Q_1$ is lower bounded by the probability of $A^*$, i.e.,
\begin{equation}\label{eqn:claim_1_1}
\mathbb{P}_{\bv} (\tilde{A} \cap Q_1) \geq \mathbb{P}_{\bv} (A^*).    
\end{equation}
Moreover, equality can only occur if $\tilde{A} \subseteq Q_1$.
\end{claim}
\begin{proof}
To see why this is the case, notice that since $\tilde{x}$ satisfies the fairness constraint, we should have
\begin{equation}
\mathbb{P}_{\bv}(\tilde{A} \cap Q_1) + \mathbb{P}_{\bv}(Q_2) \geq 
\mathbb{P}_{\bv}(\tilde{A} \cap Q_1) + \mathbb{P}_{\bv}(Q_2 \backslash \tilde{A}) \geq 
\alpha_2.    
\end{equation}
Now, if \eqref{eqn:claim_1_1} does not hold, then we would have
\begin{equation}
\mathbb{P}_{\bv} (A^*) + \mathbb{P}_{\bv}(Q_2) > \alpha_2,    
\end{equation}
but then this would mean that we can lower the $\gamma$ in allocation $x^*$ which contradicts its definition.
\end{proof}
Notice that the loss of allocation $x^*$ compared to the optimal unconstrained allocation is given by
\begin{equation} \label{eqn:loss_x_star}
\text{Loss}_{x^*} := \int_{A^*} (\phi_1(\max_{k} v_{1,k}) - \phi_2(\max_{k} v_{2,k}) ) dF(\bv)      
\end{equation}
Let also $\text{Loss}_{\tilde{x}} $ denote the loss of allocation $\tilde{x}$ compared to the optimal unconstrained allocation. This loss is lower bounded by
\begin{equation} \label{eqn:loss_x_tilde}
\text{Loss}_{\tilde{x}} \geq \int_{\tilde{A} \cap Q_1} (\phi_1(\max_{k} v_{1,k}) - \phi_2(\max_{k} v_{2,k}) ) dF(\bv).      
\end{equation}
Notice that, since $\tilde{x}$ is the optimal allocation, the right hand side of \eqref{eqn:loss_x_tilde} should be lower than \eqref{eqn:loss_x_star}: 
\begin{equation} \label{eqn:proof_static_3}
\int_{\tilde{A} \cap Q_1} (\phi_1(\max_{k} v_{1,k}) - \phi_2(\max_{k} v_{2,k}) ) dF(\bv)   \leq 
\int_{A^*} (\phi_1(\max_{k} v_{1,k}) - \phi_2(\max_{k} v_{2,k}) ) dF(\bv).    
\end{equation}
We next claim that this implies that $A^* = \tilde{A} \cap Q_1$.
\begin{claim}
We have $A^* = \tilde{A} \cap Q_1$ (up to a measure-zero set).     
\end{claim}
\begin{proof}
To simplify the notation, let $L(v_1, v_2) = \phi_1(v_1) - \phi_2(v_2)$. 
Suppose these two sets are not equal. This implies that there is some region $\tilde{B}$ outside region $A^*$ that is in region $\tilde{A} \cap Q_1$ and some region $B^*$ that is outside region $\tilde{A} \cap Q_1$ but inside region $A^*$. By Claim \ref{claim:static_proof_1}, we have:
\begin{equation} \label{eqn:probabilities_Bs}
\int_{\tilde{B}} dF(\bv) \geq  \int_{B^*} dF(\bv).   
\end{equation}
Next, by \eqref{eqn:proof_static_3}, we have 
\begin{equation}
    \label{eqn:lossrelationoptimalityproof}
    \begin{aligned}
        \int_{\tilde{B}} L(v_1, v_2) dF(\bv) \leq \int_{B^*} L(v_1, v_2) dF(\bv).
    \end{aligned}
\end{equation}
Notice that $L(v_1, v_2)$ is nonnegative over $\tilde{B}$ and $B^*$. Also, the fact that $\tilde{B} \cap A^*$ is empty (along with $\tilde{B} \subseteq Q_1$) means that $\sup_{B^*} L(v_1, v_2) \leq \gamma < \inf_{\tilde{B}} L(v_1, v_2)$. Next, notice that
\begin{equation*}
    \begin{aligned}
        \int_{B^*} L(v_1, v_2) dF(\bv) \leq \sup_{B^*} L(v_1, v_2) \int_{B^*} dF(\bv)
    \end{aligned}
\end{equation*}
and 
\begin{equation*}
    \begin{aligned}
        \int_{\tilde{B}} L(v_1, v_2)  dF(\bv) \geq \inf_{\tilde{B}} L(v_1, v_2) \int_{\tilde{B}} dF(\bv). 
    \end{aligned}
\end{equation*}
Combined with \eqref{eqn:probabilities_Bs}, these inequalities yield that 
    \begin{equation*}
        \begin{aligned}
            \int_{B^*} L(v_1, v_2) dF(\bv) \leq \int_{\tilde{B}} L(v_1, v_2) dF(\bv),
        \end{aligned}
    \end{equation*}
where the equality holds only if $\tilde{B}$ and $B^*$ are measure zero, which should be the case given \eqref{eqn:lossrelationoptimalityproof}. 
This shows that $\tilde{x}$ and $x^*$ are the same, up to a measure zero set. 
\end{proof}
This claim along with the equality condition of Claim \ref{claim:static_proof_1} shows that $A^*$ and $\tilde{A}$ are equal up to a measure-zero set. Hence, for a given $\mathcal{S}$, the optimal allocation is an affine allocation.

Now, this result was for a given $\mathcal{S}$. Using a similar argument, we can establish that the optimal set $\mathcal{S}$ should be in the form of 
\begin{equation}
\left \{ (\phi_1(v_1), \phi_2(v_2)) \Big \vert \phi_1(v_1) \geq -\eta_1 \text{ or } \phi_2(v_2) \geq -\eta_2 \right \}    
\end{equation}
for some $\eta_1 , \eta_2 \geq 0$.  We finally establish that $\eta_2 = \eta_1 + \gamma$.
Suppose that $\eta_2 < \eta_1 + \gamma$ in the optimal fair allocation. Then, there exists the following region of allocations to group one: 
$G = \{(\phi_1(v_1), \phi_2(v_2)) : \phi_2(v_2) \in (-\eta_2, -\eta_1 - \gamma), \phi_1(v_1) \in (-\eta_1, -\eta_1 - \gamma + \eta_2), \phi_2(v_2) \geq \phi_1(v_1) + \gamma \}$.  Notice that $G$  constitutes a triangle as depicted in the red region in Figure \ref{fig:delta_eps_reln_case1}. Consider modifying the allocation rule as follows. For some $g \in G$, do not allocate $g$. Instead, allocate a region of equal measure to group one beginning at $(-\eta_1, -\eta_2)$. This switch changes the loss by the amount $\eta_1 - \eta_2$. Then, we return to the original and total levels of allocation by reallocating another equal-measure region on the border $\phi_1(v_1) = \phi_2(v_2) - \gamma$ to group two. The second switch changes the loss by the amount $\gamma$. Therefore, the total change to the loss is $\eta_1 - \eta_2 + \gamma < 0$ by our initial assumption. Thus, this modification results in decreased loss to the seller, and it cannot be optimal.
\begin{figure}[t]
\centering
\begin{tikzpicture}[>=stealth]

\def\leneps{1.25}
\def\xmax{4}
\def\xmin{-4}
\def\ymax{4}
\def\ymin{-4}
\def\etaone{1}
\pgfmathsetmacro\etatwo{\etaone + \leneps}
\def\etatwooffset{0.5}
\pgfmathsetmacro\etatwoprime{\etaone + \leneps - \etatwooffset}

    \draw[->] (\xmin,0) -- (\xmax,0) node[right]{$\phi_1(v_1)$};
    \draw[->] (0,\ymin) -- (0,\ymax) node[above]{$\phi_2(v_2)$};

    \draw[-, white, name path = boundary1] (\xmin,\ymin) -- (\xmin,-\leneps);
    \draw[-, white, name path = boundary2] (\xmin,-\leneps)--(\xmin,0);
    \draw[-, white, name path = boundary3] (\xmin,0)--(\xmin,\ymax);
    \draw[-, white, name path = boundary4] (\xmin,\ymax)--(\xmax,\ymax);
    \draw[-, white, name path = boundary5] (\xmax,\ymax)--(\xmax,\ymax-\leneps);
    \draw[-, white, name path = boundary6] (\xmax,\ymax-\leneps)--(\xmax,\ymin)--(0,\ymin);
    \draw[-, white, name path = boundary7] (0,\ymin)--(\xmin,\ymin);

    \filldraw (-\etaone,0) circle (1pt);
    \node[above] at (-\etaone, 0) {$-\eta_1$};
    \filldraw (0, -\etatwoprime) circle (1pt);
    \node[right] at (0, -\etatwoprime) {$-\eta_2$};
    \filldraw (0, -\leneps) circle (1pt);
    \node[right] at (0, -\leneps) {$-\gamma$};
    \filldraw (0, -\etatwo) circle (1pt);
    \node[right] at (0, -\etatwo) {$-(\eta_1 + \gamma)$};
    
    \draw[dashed, ,->, black, name path = x_eq_y_pos] (0,0) -- (4,4) node[right]{$\phi_2(v_2) = \phi_1(v_1)$};
    \draw[dashed, ,->, black, name path = y_eq_0_neg] (0,0) -- (-\ymin,0);

    \draw[dashed, ,->, black, name path = x_eq_y_mineps_pos] (-\etaone,-\etatwo) -- (4,4-\leneps) node[right]{$\phi_1(v_1) = \phi_2(v_2) + \gamma$};

    \draw[dashed, ,->, lightgray, name path = x_eq_y_mineps_pos_neg] (-\etaone,-\etatwo) -- (-3,-3-\leneps);

    \draw[dashed, ,->, black, name path = y_eq_min_eps_neg] (0,-\etatwoprime) -- (\xmin,-\etatwoprime);

    \draw[dashed, ,->, black, name path = x_eq_neg_delta1] (-\etaone,0) -- (-\etaone,-4);

    \draw[dashed, ,->, black, name path = y_eq_neg_eps] (0, -\leneps) -- (-\xmax,-\leneps);

    \fill[green, opacity=0.2] (0, 0) -- (\xmax, \ymax) -- (\xmax, \ymax - \leneps) -- (0, -\leneps) -- (-\xmax, -\leneps) -- (-\xmax, 0) -- cycle;

    \fill[blue, opacity=0.2] (0, 0) -- (-\xmax, 0)-- (-\xmax, \ymax) -- (0, \ymax) -- (\xmax, \ymax) -- (0, 0) -- cycle;

    \fill[orange, opacity=0.2] (-\etaone, -\etatwo) -- (-\etaone, -\ymax) -- (0, -\ymax)-- (0, -\leneps) -- cycle;

    \fill[red, opacity=0.2] (-\etaone, -\etatwo) -- (-\etaone + \etatwooffset, -\etatwoprime) -- (-\etaone, -\etatwoprime) -- cycle;

    \fill[yellow, opacity=0.2](0, -\leneps) -- (\xmax, \ymax - \leneps) -- (\xmax, -\ymax) -- (0, -\ymax)-- cycle;

    \fill[green, opacity=0.2](0, -\leneps) -- (-\xmax, -\leneps) -- (-\xmax, -\etatwoprime) -- (-\etaone + \etatwooffset , -\etatwoprime ) -- cycle;

\end{tikzpicture}
\caption{The case where $\eta_2 < \eta_1 + \gamma$.}
\label{fig:delta_eps_reln_case1}
\end{figure}
Now, suppose that $\eta_2 > \eta_1 + \gamma$ in the optimal fair allocation. Then, there exists a region that is not allocated under the optimal mechanism defined by $H = \{(\phi_1(v_1), \phi_2(v_2)) : \phi_2(v_2) \in (-\eta_1 - \gamma, -\eta_1), \phi_1(v_1) \in (-\eta_2 - \gamma, -\eta_1), \phi_2(v_2) \leq \phi_1(v_1) + \gamma \}$. As before, $H$ is a triangle depicted in gray in Figure \ref{fig:delta_eps_reln_case2}. However, notice that the unallocated values in $H$ are closer to the unconstrained allocation boundary than the allocated set of virtual values $\{(\phi_1(v_1), \phi_2(v_2)): \phi_1(v_1) < =\eta_1, \phi_2(v_2) \in (-\eta_1 - \gamma, -\eta_2),  \phi_2(v_2) \geq \phi_1(v_1) - \gamma$ \}. Therefore, such an allocation cannot be the optimal fair allocation due to the result shown in the proof of Theorem \ref{thrm:fair_allocation_static} that loss is increasing in Euclidean distance from the unconstrained allocation boundary (\ref{fig:allocation_static}).
\begin{figure}[t]
\centering
\begin{tikzpicture}[>=stealth]

\def\leneps{1.25}
\def\xmax{4}
\def\xmin{-4}
\def\ymax{4}
\def\ymin{-4}
\def\etaone{1}
\pgfmathsetmacro\etatwo{\etaone + \leneps}
\def\etatwooffset{0.5}
\pgfmathsetmacro\etatwoprime{\etaone + \leneps +\etatwooffset}

    \draw[->] (\xmin,0) -- (\xmax,0) node[right]{$\phi_1(v_1)$};
    \draw[->] (0,\ymin) -- (0,\ymax) node[above]{$\phi_2(v_2)$};

    \draw[-, white, name path = boundary1] (\xmin,\ymin) -- (\xmin,-\leneps);
    \draw[-, white, name path = boundary2] (\xmin,-\leneps)--(\xmin,0);
    \draw[-, white, name path = boundary3] (\xmin,0)--(\xmin,\ymax);
    \draw[-, white, name path = boundary4] (\xmin,\ymax)--(\xmax,\ymax);
    \draw[-, white, name path = boundary5] (\xmax,\ymax)--(\xmax,\ymax-\leneps);
    \draw[-, white, name path = boundary6] (\xmax,\ymax-\leneps)--(\xmax,\ymin)--(0,\ymin);
    \draw[-, white, name path = boundary7] (0,\ymin)--(\xmin,\ymin);

    \filldraw (-\etaone,0) circle (1pt);
    \node[above] at (-\etaone, 0) {$-\eta_1$};
    \filldraw (0, -\etatwoprime) circle (1pt);
    \node[right] at (0, -\etatwoprime) {$-\eta_2$};
    \filldraw (0, -\leneps) circle (1pt);
    \node[right] at (0, -\leneps) {$-\gamma$};
    \filldraw (0, -\etatwo) circle (1pt);
    \node[right] at (0, -\etatwo) {$-(\eta_1 + \gamma)$};
    
    \draw[dashed, ,->, black, name path = x_eq_y_pos] (0,0) -- (4,4) node[right]{$\phi_2(v_2) = \phi_1(v_1)$};
    \draw[dashed, ,->, black, name path = y_eq_0_neg] (0,0) -- (-\ymin,0);

    \draw[dashed, ,->, black, name path = x_eq_y_mineps_pos] (-\etaone-\etatwooffset,-\etatwo-\etatwooffset) -- (4,4-\leneps) node[right]{$\phi_1(v_1) = \phi_2(v_2) + \gamma$};

    \draw[dashed, ,->, lightgray, name path = x_eq_y_mineps_pos_neg] (-\etaone-\etatwooffset,-\etatwo-\etatwooffset) -- (-3,-3-\leneps);

    \draw[dashed, ,->, black, name path = y_eq_min_eps_neg] (0,-\etatwoprime) -- (\xmin,-\etatwoprime);

    \draw[dashed, ,->, black, name path = x_eq_neg_delta1] (-\etaone,0) -- (-\etaone,-4);

    \draw[dashed, ,->, black, name path = y_eq_neg_eps] (0, -\leneps) -- (-\xmax,-\leneps);

    \fill[green, opacity=0.2] (0, 0) -- (\xmax, \ymax) -- (\xmax, \ymax - \leneps) -- (0, -\leneps) -- (-\xmax, -\leneps) -- (-\xmax, 0) -- cycle;

    \fill[blue, opacity=0.2] (0, 0) -- (-\xmax, 0)-- (-\xmax, \ymax) -- (0, \ymax) -- (\xmax, \ymax) -- (0, 0) -- cycle;

    \fill[orange, opacity=0.2] (-\etaone, -\etatwo) -- (-\etaone, -\ymax) -- (0, -\ymax)-- (0, -\leneps) -- cycle;

    \fill[lightgray, opacity=0.2] (-\etaone, -\etatwo) -- (-\etaone - \etatwooffset, -\etatwoprime) -- (-\etaone, -\etatwoprime) -- cycle;

    \fill[yellow, opacity=0.2](0, -\leneps) -- (\xmax, \ymax - \leneps) -- (\xmax, -\ymax) -- (0, -\ymax)-- cycle;

    \fill[green, opacity=0.2](0, -\leneps) -- (-\xmax, -\leneps) -- (-\xmax, -\etatwoprime) -- (-\etaone - \etatwooffset , -\etatwoprime ) -- cycle;

\end{tikzpicture}
\caption{The case where $\eta_2 > \eta_1 + \gamma$.}
\label{fig:delta_eps_reln_case2}
\end{figure}
Observing that these two cases together imply equality concludes the proof. $\blacksquare$
\subsection{Proof of \cref{proposition:optimal_static}}
Let $\text{Loss}_{x}$ and $\text{Loss}_{\widehat{x}}$ denote the loss of allocations $x$ and $\Tilde{x}$ compared to the original unconstrained allocation, respectively. Notice that the change in loss under the modified allocation can be expressed as 
\begin{equation}
    \begin{aligned}
        \text{Loss}_{x} - \text{Loss}_{\widehat{x}} &= \int_{\widehat{G}_2} (\phi_1(\max_{k} v_{1,k}) - \phi_2(\max_{k} v_{2,k}) ) dF(\bv).
    \end{aligned}
\end{equation}
Since $\phi_1(\max_{k} v_{1,k}) - \phi_2(\max_{k} v_{2,k}) \in (0, \gamma)$ for every pair $(\max_{k} v_{1,k}, \max_{k} v_{2,k}) \in \widehat{G}_2$ and $\widehat{G}_2$ is a set with nonzero measure, we may conclude that $\text{Loss}_{x} - \text{Loss}_{\widehat{x}} > 0$ In other words, the seller's utility increases under the allocation $\widehat{x}$ as compared to that of allocation $x$. Therefore, for an optimal fair allocation where $\gamma > 0$, it must be that $\mathbb{P}(G_2) = \alpha_2$.

The proof is similar to show that, if $\eta_1 >0$, i.e., there is insufficient total allocation in the optimal unconstrained allocation, we have $\mathbb{P}(G_1) = \alpha_1$. As before, clearly $\mathbb{P}(G_1) \geq \alpha_1$, otherwise the fairness constraint is not satisfied. Suppose, for sake of contradiction, that $\mathbb{P}(G_1) \geq \alpha_1$. Then there is some subset 
\begin{equation}
    \begin{aligned}
        \widehat{G}_1 \subseteq \{(v_1, v_2)|-\eta_1 \leq \phi_1(v_1), \leq 0, \phi_2(v_2) \leq \phi_1(v_1) - \gamma\}
    \end{aligned}
\end{equation}
with nonzero measure such that $\mathbb{P}(G_1 \setminus \widehat{G}_2) \geq \alpha_1$. Notice that $\widehat{G}_1$ is some subset of the orange region in Figure \ref{fig:fair_allocation_static}. Once again, let $x$ be the original allocation and $\widehat{x}$ be the modified allocation where we do not allocate region $\widehat{G}_2$. The change in loss under the modified allocation is given by \begin{equation}
    \begin{aligned}
        \text{Loss}_{x} - \text{Loss}_{\widehat{x}} &= \int_{\widehat{G}_1} (\phi_1(\max_{k} v_{1,k}) - \phi_2(\max_{k} v_{2,k}) ) dF(\bv).
    \end{aligned}
\end{equation}
Since $\widehat{G}_1$ has nonzero measure and $\phi_2(v_2) \leq 0$ for all $(\max_{k} v_{1,k}, \max_{k} v_{2,k}) \in \widehat{G}_1$, it follows that $\text{Loss}_{\widehat{x}} - \text{Loss}_{x} > 0$. Once again, this contradicts the optimality of allocation $x$, and it must be that $\mathbb{P}(G_1) = \alpha_1$.

Observing that the second claim follows from exchanging the roles of groups 1 and 2 in the preceding work concludes the proof. $\blacksquare$
\subsection{Proof of \cref{corollary:last_round}}
Both results follow from an application of the following claim. 
\begin{claim}\label{claim:expect_x_times_f}
Let $f(v_{i,k}^T)$ be some function of $v_{i,k}^T$. Then
    \begin{equation}
        \begin{aligned}
            \mathbb{E}\left[x_{i,k}^T(\bb^T, \bh^T)  f(v_{i,k}^T)\right] &= \frac{1}{n}\int_{G_i} f(v_{i,k}^T) d F_\text{max}^T(v_1,v_2)
        \end{aligned}
    \end{equation}
\end{claim}
First, by the law of total probability and observing that $x_{i,k}^T = 0$ if $v_{i,k}^T < \max_k v_{i,k}^T$, 
\begin{equation}
    \begin{aligned}
        \mathbb{E}\left[x_{i,k}^T(\bb^T, \bh^T)  f(v_{i,k}^T)\right] &= \mathbb{E}\left[x_{i,k}^T(\bb^T, \bh^T)  f(v_{i,k}^T) \mid v_{i,k}^T = \max_k v_{i,k}^T\right] \mathbb{P}\left(v_{i,k}^T = \max_k v_{i,k}^T\right)
    \end{aligned}
\end{equation}
Next, since $v_{i,k}^T \overset{i.i.d.}{\sim} F_{i}^T,$ $\mathbb{P}(v_{i,k}^T = \max_k v_{i,k}) = \frac{1}{n}$. Also observe that $x_{i,k}^T = 0$ for $(v_1, v_2) \notin G_i$. Thus, 
\begin{equation}
        \begin{aligned}
            \mathbb{E}\left[x_{i,k}^T(\bb^T, \bh^T)  f(v_{i,k}^T)\right] &= \frac{1}{n}\mathbb{E}\left[f(v_{i,k}^T) \mathbbm{1}\{v_{i,k}^T = \max_k v_{i,k}^T, (v_1, v_2) \in G_i\}\right]\\
            &=\frac{1}{n}\int_{G_i} f(v_{i,k}^T) d F_\text{max}^T(v_1,v_2).
        \end{aligned}
    \end{equation}
Now, we use the preceding claim to evaluate the expected utility of a buyer in group $i$ denoted $\nu_i^T(R_1^T, R_2^T)$. From \eqref{eqn:buyer_k_utility_t},
\begin{equation}\label{eqn:buyer_utility_T}
    \begin{aligned}
        \nu_i^T(R_1^T, R_2^T)=\mathbb{E} \left[ \mathcal{U}_{i,k}^T(v_{i,k}^T; \bb^T, \bh^T)\right]=\mathbb{E} \left[ x_{i,k}^T(\bb^T, \bh^T) v_{i,k}^T - p_{i,k}^T(\bb^T, \bh^T)\right].
    \end{aligned}
\end{equation}
The expected payment of a buyer in group $i$ can equivalently be expressed as
\begin{equation}\label{eqn:paymentequiv}
    \begin{aligned}
        \mathbb{E}\left[p_{i,k}^T(\bb^T, \bh^T)\right]=\mathbb{E}\left[x_{i,k}^T(\bb^T, \bh^T)  \phi_i^T(v_{i,k}^T)\right]
    \end{aligned}
\end{equation}
as shown in \cite[Lemma~13.11]{AGT}. Substitution of this result into \cref{eqn:buyer_utility_T} yields 
\begin{equation}\label{eqn:expec_util_buyer}
    \begin{aligned}
        \nu_i^T(R_1^T, R_2^T)=\mathbb{E} \left[ x_{i,k}^T(\bb^T, \bh^T) \left(v_{i,k}^T - \phi_i^T(v_{i,k}^T)\right)\right].
    \end{aligned}
\end{equation}
From here, applying Claim \ref{claim:expect_x_times_f} to the function $v_{i,k}^T - p_{i,k}^T(\bb^T, \bh^T)$ yields 
the desired result. 

Next, we show the second part of the claim. By Definition \ref{eqn:seller_utility_static}, we may express the expected utility of the seller in group $i$ denoted $\mu^T(R_1^T, R_2^T)$ as
\begin{equation}
    \begin{aligned}
        \mathbb{E} \left[\phi_{1}(v_{1,k}) x_{1,k}(\bv) + \phi_{2}(v_{2,k}) x_{2,k}(\bv)\right].
    \end{aligned}
\end{equation}
By linearity of expectation and an application of Claim \ref{claim:expect_x_times_f} to functions $\phi_{1}(v_{1,k})$ and $\phi_{2}(v_{2,k})$ yields the desired result. $\blacksquare$
\subsection{Proof of \cref{proposition:dynamic_one_group}} 
First, notice that the seller must allocate the item to a member of the group, otherwise the auction becomes infeasible. Therefore, this case reduces to a static Vickrey auction with no reserve price where the revenue-maximizing allocation and payments are characterized by \cite{myerson1981optimal} as follows. To optimize seller revenue, the seller allocates to the highest bidder who pays an amount equal to the second highest bid. 

By definition, we may express the expected utility as
\begin{equation}
    \begin{aligned}
        \nu_{i}^t(R_1^t, R_2^t) = \mathbb{E}\left[\sum_{\tau = t}^{T} \delta^{\tau - t} \mathcal{U}_{i, k}^{\tau}\right] = \mathbb{E}\left[\mathcal{U}_{i, k}^{t}\right] + \delta\mathbb{E}\left[\sum_{\tau = t + 1}^{T} \delta^{\tau - (t + 1)} \mathcal{U}_{i, k}^{\tau}\right].
    \end{aligned}
\end{equation}
Since the seller allocates the item to a buyer in group $i$ with probability 1 to maintain feasibility and by Fact \ref{fact:resid_min_alloc_update}, 
\begin{equation}
    \begin{aligned}
      \delta\mathbb{E}\left[\sum_{\tau = t + 1}^{T} \delta^{\tau - (t + 1)} \mathcal{U}_{i, k}^{\tau}\right] = \delta \nu_{i}^{t+1}((R_{i}^t - 1)/\delta, R_{-i}^t/\delta).
    \end{aligned}
\end{equation}
From \eqref{eqn:expec_util_buyer}, \eqref{eqn:virtual_value}, and the observation that this setting reduces to a Vickrey auction: 
\begin{equation}
    \begin{aligned}
        \mathbb{E} \left[ x_{i,k}^t(\bb^t, \bh^t) \left(v_{i,k}^t - \phi_i^t(v_{i,k}^t)\right)\right] = \mathbb{E} \left[ \frac{1-F_i^t(v_{i,k}^t)}{f_i^t(v_{i,k}^t)} \mathbbm{1}\left\{ v_{i,k}^t = \max_j v_{i,j}^t\right\}\right].
    \end{aligned}
\end{equation}
Since $v_{i,k}^t$ is distributed according to $F_i^t(v_{i,k}^t)$ for all $k \in [n]$, $\mathbb{P}\left(v_{i,k}^t = \max_j v_{i,j}^t \right) = (F_{i}^t(v))^{n-1}$. Therefore, 
\begin{equation}
    \begin{aligned}
        \mathbb{E} \left[ x_{i,k}^t(\bb^t, \bh^t) \left(v_{i,k}^t - \phi_i^t(v_{i,k}^t)\right)\right]  = \int (1-F_{i}^t(v))(F_{i}^t(v))^{n-1} dv
    \end{aligned}
\end{equation}
and the result follows. 

Next, we evaluate the expected utility of buyer $-i$ that does not receive the item. Since buyer $-i$ does not receive the item with probability 1, their expected utility reduces to the discounted expected utility over future rounds, i.e., 
\begin{equation}
    \begin{aligned}
        \nu_{-i}^t(R_1^t, R_2^t) =  \delta\mathbb{E}\left[\sum_{\tau = t + 1}^{T} \delta^{\tau - (t + 1)} \mathcal{U}_{i, k}^{\tau}\right].
    \end{aligned}
\end{equation}
The result follows by applying Fact \ref{fact:resid_min_alloc_update} to see that 
\begin{equation}
    \begin{aligned}  \delta\mathbb{E}\left[\sum_{\tau = t + 1}^{T} \delta^{\tau - (t + 1)} \mathcal{U}_{i, k}^{\tau}\right] = \delta \nu_{-i}^{t+1}((R_1^t - 1)/\delta, R_2^t/\delta).
    \end{aligned}
\end{equation}

Last, we evaluate the seller's expected utility. By Definition \ref{eqn:seller_utility_static} and since the good is allocated to the maximum valuation buyer in group $i$ with probability 1, 
\begin{equation}
    \begin{aligned}
        \mu^t(R_1^t, R_2^t) = \mathbb{E}\left[\sum_{\tau = t}^{T} \delta^{\tau - t} \left [ \sum_{i \in [2],k \in [n]} \phi^{\tau}_{i}(v^{\tau}_{i,k}) x^{\tau}_{i,k}(\bv) \right] \right] = \mathbb{E} \left[ \phi_{i}^{t}\left(\max_{k} v^{t}_{i,k}\right)\right] + \delta \mu^{t+1}((R_i^t - 1)/\delta, R_{-i}^t/\delta).
    \end{aligned}
\end{equation}

Recalling that $\max_{k} v^{t}_{i,k} \sim F_i^t(v)^n$, we can express the first term as 
\begin{equation}
    \begin{aligned}
        \mathbb{E} \left[ \phi_{i}^{t}\left(\max_{k} v^{t}_{i,k}\right)\right] =  \int \phi_i^t(v) d F_i^t(v)^n
    \end{aligned}
\end{equation}
and the result follows. $\blacksquare$
\subsection{Proof of \cref{theorem:dynamic_both_feasible}} \label{proof:theorem:dynamic_both_feasible}
Notice that the variable 
\begin{equation}
\sum_{\ell=1}^n x_{i,\ell}^t(\bv^t).    
\end{equation}
determines whether group $i$ receives the item or not. Now, suppose buyer $(i,k)$ submits bid $b$. Given Assumption \ref{assumption:no_reserve_price}, we can write the utility of buyer $(i,k)$ as:
\begin{equation}\label{eqn:dynamic_proof_1}
x_{i,k}^t(b, \bv_{-(i,k)}^t) v_{i,k}^t - p_{i,k}^t(b, \bv_{-(i,k)}^t)
- \delta \left(\sum_{\ell=1}^n x_{i,\ell}^t(b, \bv_{-(i,k)}^t) \right) \Delta_i^t(R_1^t, R_2^t) + 
\delta \nu_i^{t+1}(R_i^t/\delta, (R_{-i}^t-1)/\delta).
\end{equation}
Notice that the last term $\delta \nu_i^{t+1}(R_i^t/\delta, (R_{-i}^t-1)/\delta)$ is a constant, independent of the values and the bid. Hence, we can drop this constant from the buyer's decision making. Now, assuming $\bv_{-(i,k)}^t$ is fixed, we can write the adjusted utility as 
\begin{equation} \label{eqn:dynamic_proof_2}
x_{i,k}^t(b, \bv_{-(i,k)}^t) v_{i,k}^t - \tilde{p}_{i,k}^t(b, \bv_{-(i,k)}^t),
\end{equation}
where 
\begin{equation} \label{eqn:dynamic_proof_3}
\tilde{p}_{i,k}^t(b, \bv_{-(i,k)}^t) := p_{i,k}^t(b, \bv_{-(i,k)}^t)
+ \delta (\sum_{\ell=1}^n x_{i,\ell}^t(b, \bv_{-(i,k)}^t)) \Delta_i^t(R_1^t, R_2^t).   
\end{equation}
Using this representation, we can interpret round $t$-th auction as a one round auction. Hence, given \cref{proposition:static_EPIC}, for an allocation to be EPIC, $x_{i,k}^t(b, \bv_{-(i,k)}^t) $ should be (weakly) monotone in $b$ and we also should have 
\begin{equation}
\tilde{p}_{i,k}^t(b, \bv_{-(i,k)}^t) = b~x_{i,k}^t(b, \bv_{-(i,k)}^t) - \int_{\underline{v}_i^t}^b x_{i,k}^t(z, \bv_{-(i,k)}^t) dz + c_{i,k}(\bv_{-(i,k)}^t),   
\end{equation}
for some function $c_{i,k}(\cdot)$ which is independent of $v$. As a result, for any $i \in [2]$ and $k \in [n]$, the payment is given by
\begin{align} \label{eqn:dynamic_proof_4}
& p_{i,k}^t(v, \bv_{-(i,k)}^t)  = \\
& v~x_{i,k}^t(v, \bv_{-(i,k)}^t) - \int_{\underline{v}_i^t}^v x_{i,k}^t(z, \bv_{-(i,k)}^t) dz 
- \delta (\sum_{\ell=1}^n x_{i,\ell}^t(v, \bv_{-(i,k)}^t)) \Delta_i^t(R_1^t, R_2^t) + c_{i,k}(\bv_{-(i,k)}^t).  \nonumber  
\end{align}
The seller wants to choose $c_{i,k}(\cdot)$ as large as possible, but the IR constraint puts an upper bound on it. Substituting the payment \eqref{eqn:dynamic_proof_4} into \eqref{eqn:dynamic_proof_1}, the buyer utility is given by
\begin{equation}\label{eqn:dynamic_proof_5}
\int_{\underline{v}_i^t}^{v_{i,k}^t} x_{i,k}^t(z, \bv_{-(i,k)}^t) dz
+ \delta \nu_i^{t+1}(R_i^t/\delta, (R_{-i}^t-1)/\delta) - c_{i,k}(\bv_{-(i,k)}^t).
\end{equation}
Now, to check the IR constraint, suppose buyer $(i,k)$ skips round $t$, and denote the optimal allocation in that case by $\tilde{\bx}^t$. Hence, the expected utility of buyer $(i,k)$ from time $t+1$ onwards is given by
\begin{equation}\label{eqn:dynamic_proof_6}
 \delta \nu_i^{t+1}(R_i^t/\delta, (R_{-i}^t-1)/\delta) - \delta (\sum_{\ell \neq k} \tilde{x}_{i,\ell}^t(\bv_{-(i,k)}^t)) \Delta_i^t(R_1^t, R_2^t).
\end{equation}
The IR constraint implies that the expectation of \eqref{eqn:dynamic_proof_5} minus \eqref{eqn:dynamic_proof_6} over $\bv_{-(i,k)}^t$ should be nonnegative for any $v_{i,k}^t$. Thus, we should have 
\begin{equation}\label{eqn:IR_proof_i}
\mathbb{E}_{\bv_{-(i,k)}^t}[c_{i,k}(\bv_{-(i,k)}^t)] \leq \delta \Delta_i^t(R_1^t, R_2^t) ~\mathbb{E}_{\bv_{-(i,k)}^t}\left[ \sum_{\ell \neq k} \tilde{x}_{i,\ell}^t(\bv_{-(i,k)}^t) \right],
\end{equation}
and so, the seller sets the function $c_{i,k}(\cdot)$ to achieve the equality case.
Note that the right hand side is only a function of $i$. We denote it by $c_i^t$.
We next make the following claim regarding the expected payment of buyer $(i,k)$ to the seller in this case.
\begin{claim} \label{claim:expected_payment_dynamic_1}
The expected payment of buyer $(i,k)$ to the seller at round $t$ is given by
\begin{equation}
\mathbb{E}_{\bv^t} \left[ \phi_i^t(v_{i,k}^t) x_{i,k}^t(\bv^t) 
- \delta (\sum_{\ell=1}^n x_{i,\ell}^t(\bv^t)) \Delta_i^t(R_1^t, R_2^t)
\right ] + c_i^t.
\end{equation}
\end{claim}
The proof of this claim follows from the classical derivation of payment in static mechanism design (see \cite{AGT}[Lemma 13.11] for instance) along with the characterization of payment in \eqref{eqn:dynamic_proof_4}. As a consequence, the seller's utility at round $t$ is given by 
\begin{equation} \label{eqn:seller_utility_dynamic_1}
\mathbb{E}_{\bv^t} \left[ \sum_{i=1}^2 \sum_{k=1}^n \phi_i^t(v_{i,k}^t) x_{i,k}^t(\bv^t) 
- \sum_{i=1}^2 n \delta \Delta_i^t(R_1^t, R_2^t) \sum_{\ell=1}^n x_{i,\ell}^t(\bv^t)\right ]  + nc_1^t + n c_2^t.   
\end{equation}
Therefore, the seller's expected utility for time $t$ onwards is given by 
\begin{small}
\begin{equation} \label{eqn:seller_utility_dynamic_2}
\mathbb{E}_{\bv^t} \left[ \sum_{i=1}^2 \sum_{k=1}^n \phi_i^t(v_{i,k}^t) x_{i,k}^t(\bv^t) 
+ \sum_{i=1}^2 
\delta \left ( \mu^{t+1}((R_{i}^t-1)/\delta, R_{-i}^t/\delta) - 
n \Delta_i^t(R_1^t, R_2^t) \right )
\sum_{\ell=1}^n x_{i,\ell}^t(\bv^t) \right ]  + nc_1^t + n c_2^t.
\end{equation}
\end{small}
Notice that the seller wants to maximize \eqref{eqn:seller_utility_dynamic_2}. First, notice that, since $\phi_i^t(\cdot)$ is increasing by Assumption \ref{assumption:regularity}, if we allocate the item to group $i$, we would allocate it to the buyer with the highest value which we denote its value by $v_i^t = \max_k v_{i,k}^t$. 

Now, we need to determine whether the item is allocated to the buyer with the highest value in group one or to the buyer with the highest value in the second group. Given \eqref{eqn:seller_utility_dynamic_2}, it is straightforward to see that the item is allocated to group one if 
\begin{align}
& \phi_1^t(v_1^t) - \phi_2^t(v_2^t) \geq \\
& n \delta \left ( \Delta_1^t(R_1^t, R_2^t) - \Delta_2^t(R_1^t, R_2^t)\right ) + \delta \left( 
\mu^{t+1}((R_{2}^t-1)/\delta, R_{1}^t/\delta) - \mu^{t+1}((R_{1}^t-1)/\delta, R_{2}^t/\delta)
\right ), \nonumber
\end{align}
and otherwise it is allocated to the second group. Notice that this allocation is indeed monotone, and hence, along with the above payment, it is an EPIC and IR allocation. This gives us the set $G_i^t$ in the theorem's statement. Finally, to derive the update of interim functions, notice that the expected utility of seller follows from \eqref{eqn:seller_utility_dynamic_2}. The expected utility of buyers can also be derived from Claim \ref{claim:expected_payment_dynamic_1} similar to the proof of Corollary \ref{corollary:last_round}.
\subsubsection{Characterization of $\zeta_i^t$}
Recall that $c_i^t$ is given by
\begin{equation}
c_i^t = \delta \Delta_i^t(R_1^t, R_2^t) \zeta_i^t,
\end{equation}
where 
\begin{equation}
\zeta_i^t = \mathbb{E}_{\bv_{-(i,k)}^t}\left[ \sum_{\ell \neq k} \tilde{x}_{i,\ell}^t(\bv_{-(i,k)}^t) \right]     
\end{equation}
is the probability of group $i$ winning the object when they participate with one fewer buyer. Notice that, an argument similar to the one we made above implies that in the case where we have $n-1$ buyers from group $i$ and $n$ buyers from the other group, group $i$ wins the item if $(\max_{\ell \neq k} v_{i,\ell}^t, \max_{\ell} v_{(-i), \ell}^t) \in \tilde{G}_i^t$ with
\begin{align}\label{eqn:G_i_tilde}
\tilde{G}_i^t := \Big \{ (v_1, v_2) ~\Big \vert~  \phi_i^t(v_i) - \phi_{-i}^t(v_{-i}) \geq (n-1) \delta \Delta_i^t(R_1^t, R_2^t) - n \delta \Delta_{-i}^t(R_1^t, R_2^t) + (-1)^i \delta \Delta_0^t(R_1^t, R_2^t) \Big \}.
\end{align}
As a result, we have
\begin{equation}\label{eqn:zeta_defn}
\zeta_i^t = \int_{\tilde{G}_i^t} d(F_i^t(v_i))^{n-1} (F_{-i}^t(v_{-i}))^{n}.   
\end{equation}
\subsubsection{Relaxing Assumption \ref{assumption:no_reserve_price}} \label{sec:relax_assumption_allocation}
Revisiting our proof shows that we can relax Assumption \ref{assumption:no_reserve_price} and repeat the proof steps. In particular, when not allocating to either of the two groups is feasible, the utility of buyer $(i,k)$ in \eqref{eqn:dynamic_proof_1} will change to
\begin{align}\label{eqn:dynamic_proof_1_p}
& x_{i,k}^t(b, \bv_{-(i,k)}^t) v_{i,k}^t - p_{i,k}^t(b, \bv_{-(i,k)}^t)
+ \delta \nu_i^{t+1}(R_i^t/\delta, R_{-i}^t/\delta) \nonumber \\
& - \delta (\sum_{\ell=1}^n x_{i,\ell}^t(b, \bv_{-(i,k)}^t))\Delta'_{i,i}(R_1^t, R_2^t)
-  \delta (\sum_{\ell=1}^n x_{(-i),\ell}^t(b, \bv_{-(i,k)}^t))\Delta'_{i,-i}(R_1^t, R_2^t),
\end{align}
with 
\begin{align}
\Delta'_{i,i}(R_1^t, R_2^t) &= \nu_i^{t+1}(R_i^t/\delta, R_{-i}^t/\delta) - \nu_i^{t+1}((R_i^t-1)/\delta, R_{-i}^t/\delta), \\
\Delta'_{i,-i}(R_1^t, R_2^t)  &= \nu_i^{t+1}(R_i^t/\delta, R_{-i}^t/\delta) - \nu_i^{t+1}(R_i^t/\delta, (R_{-i}^t-1)/\delta).
\end{align}
We can follow similar steps to derive the payment and allocations. We do not repeat the proof steps here as they follow a similar argument.
\subsection{Proof of Proposition \ref{prop:seller_payment}}
\label{proof:seller_payment}
We begin by observing that, by independence, 
    \begin{equation*}
        \begin{aligned}
         \mathbb{E}\left[ n \delta \sum_{i=1}^2 \Delta_i^t(R_1^t, R_2^t) (\zeta_i^t - \mathbbm{1}(i^*=i))\right]  = n \delta \sum_{i=1}^2   \Delta_i^t(R_1^t, R_2^t) (\zeta_i^t - \mathbb{P}(G_i^t))
        \end{aligned}
    \end{equation*}
    Next, observe the following lemma. \begin{lemma}\label{lemma:zeta_minus_pgi}
        \begin{equation*}
            \begin{aligned}
                \zeta_i^t - \mathbb{P}(G_i^t) \geq \frac{-1}{n}
            \end{aligned}
        \end{equation*}
    \end{lemma}
    \begin{proof}
    By \cref{eqn:zeta_defn}
        \begin{equation*}
            \begin{aligned}
                \zeta_i^t - \mathbb{P}(G_i^t) = \int_{\tilde{G}_i^t} d(F_i^t(v_i))^{n-1} (F_{-i}^t(v_{-i}))^{n} - \int_{G_i^t} d \left(F_i^t(v_i) \right)^n \left(F_{-i}^t(v_{-i}) \right)^n. 
            \end{aligned}
        \end{equation*}
        where $G_i^t$ is defined in \cref{eqn:G_i_t}.
        Letting $h(v_i, v_{-i}) \coloneqq n^2 \left(F_i^t(v_i) \right)^{n-1} f_i^t(v_i)  \left(F_{-i}^t(v_{-i}) \right)^{n-1} f_{-i}^t(v_{-i})$, note that we may equivalently express 
        \begin{equation*}
            \begin{aligned}
                \zeta_i^t - \mathbb{P}(G_i^t) = \int_{\tilde{G}_i^t} \frac{n-1}{n} \frac{1}{F_i^t(v_i)} h(v, v_{-i}) d v_i d v_{-i} - \int_{G_i^t} h(v, v_{-i}) d v_i d v_{-i}
            \end{aligned}
        \end{equation*}
        where $\tilde{G}_i^t$ is defined in \cref{eqn:G_i_tilde}. Next, by comparing the thresholds that define regions $G_i^t$ and $\tilde{G}_i^t$ in equations \cref{eqn:G_i_t} and \cref{eqn:G_i_tilde}, we can see that 
            \begin{equation*}
                \begin{aligned}
                    G_i^t \subseteq \tilde{G}_i^t.
                \end{aligned}
            \end{equation*}
    It follows from this observation and the fact that $F_i^t(v_i) \leq 1$ for all $v_i$ that
     \begin{equation*}
            \begin{aligned}
                \zeta_i^t - \mathbb{P}(G_i^t) \geq \int_{ 
            {G}_i^t} \frac{n-1}{n} \frac{1}{F_i^t(v_i)} h(v, v_{-i}) d v_i d v_{-i} - \int_{G_i^t} h(v, v_{-i}) d v_i d v_{-i}  = \frac{-\mathbb{P}(G_i^t)}{n}  
            \end{aligned}
        \end{equation*}
    \end{proof}
    Now, by \cref{lemma:zeta_minus_pgi}, 
    \begin{equation*}
        \begin{aligned}
            \mathbb{E}\left[ \delta \Delta_i^t(R_1^t, R_2^t) (\zeta_i^t - \mathbbm{1}(i^*=i))\right] \geq -\delta \left(   \Delta_1^t(R_1^t, R_2^t) + \Delta_2^t(R_1^t, R_2^t)\right)
        \end{aligned}
    \end{equation*}
    Observe that, from here, it suffices to show that $\Delta_i^t(R_1^t, R_2^t) \leq \mathcal{O}(\frac{1}{n})$ for all $i$. Instead, we show the following stronger lemma to conclude the proof.
    \begin{lemma}\label{lemma:nu_bound}
    For any $i$ and $t$
        \begin{equation*}
            \begin{aligned}
                \nu_i^t(R_1^t, R_2^t) \leq \mathcal{O}\left(\frac{1}{n}\right)
            \end{aligned}
        \end{equation*}
    \end{lemma}
    \begin{proof}
        We demonstrate the claim by backwards induction. Recall that round $T$ is simply the static case. Let $i^*$ denote the index of the winner's group and let $k^*$ denote the index of the winner in group $i^*$. Observe that, in round $T$, a buyer $(i,k)$'s utility is given by 
        \begin{equation*}
            \begin{aligned}
                \mathbbm{1}(i^*=i)\mathbbm{1}(k^*=k) \left( v^T_{i,k} - \max_{i, -k} v^T_{i,k} \right) \leq V. 
            \end{aligned}
        \end{equation*}
        Since buyer values within group $i$ are i.i.d., it follows that 
        \begin{equation}\label{eqn:nu_bound_V}
            \begin{aligned}
                \nu_{i}^T(R_1^t, R_2^t) \leq \frac{V}{n}
            \end{aligned}
        \end{equation}
        for all $i, k$.
        Now, supposing that the claim holds for round $t+1$, observe that, in round $t$, the seller faces one of two cases: either (1) they must allocate to a buyer in a particular group such that the auction remains feasible or (2) they may allocate to either group. We show that the claim holds in both cases. 

        In case 1, combining \cref{eqn:nu_bound_V} and \cref{eqn:dynamic_one_group_updates_nu_i} yields 
        \begin{equation*}
            \begin{aligned}
                \nu_{i}^t(R_1^t, R_2^t) \leq \frac{V}{n} + \delta \nu_{i,k}^{t+1}(R_1^t, R_2^t)
            \end{aligned}
        \end{equation*}
        and the claim follows by our inductive hypothesis. 

        In case 2, combining \cref{eqn:nu_bound_V} and \cref{eqn:dynamic_two_groups_updates_nu} yields
        \begin{equation*}
            \begin{aligned}
                \nu_{i}^t(R_1^t, R_2^t) \leq \frac{V}{n} + \delta \nu_{i,k}^{t+1}(R_1^t, R_2^t) - \Delta_i^t(R_1^t, R_2^t) \zeta_i^t \leq \frac{V}{n} + \delta \nu_{i,k}^{t+1}(R_1^t, R_2^t)
            \end{aligned}
        \end{equation*}
        where the second inequality follows from Assumption \ref{assumption:pos_Delta}.

    It follows that, in either case, 
    \begin{equation*}
        \begin{aligned}
            \nu_{i}^t(R_1^t, R_2^t) \leq \frac{V}{n} \sum_{\tau = 0}^{T-t} \delta^{\tau} \leq \frac{V}{n} \frac{1}{1-\delta}
        \end{aligned}
    \end{equation*}
    \end{proof}
    $\blacksquare$
\subsection{Relaxation of Assumption \ref{assumption:regularity} by Ironing}
\label{proof:ironing}
Here we demonstrate that the regularity assumption can be relaxed using the \emph{ironing} technique \cite{myerson1981optimal}. For $i \in [2]$, let $h_i(.)$ be the virtual value function in the quantile space, i.e., $h_i(q) = \phi_i(F_i^{-1}(q))$ for any $q \in [0,1]$. Also, let $H_i$ be its cumulative virtual value function, i.e, $H_i(q) = \int_0^q h_i(q') dq'$. Now, let us recall the ironing technique from \cite{myerson1981optimal}. 

We define $G_i: [0,1] \to \mathbb{R}$ as the convex hull of $H_i$, i.e., the largest convex function underestimator of $H_i$, and denote its derivative by $g'_i(\cdot)$. Now, the ironed virtual value function $\tilde{\phi}_i(\cdot)$ is simply defined as $\tilde{\phi}_i(v) = g_i(F_i(v))$. Notice that, given that we have dropped Assumption \ref{assumption:regularity}, $\phi_i(\cdot)$ is not necessarily monotone. However, given the ironing procedure, $\tilde{\phi}_i(\cdot)$ is monotone. In fact, when Assumption \ref{assumption:regularity} holds, i.e., for regular distributions, $H_i(\cdot)$ is convex itself, and hence $g_i = h_i$ which implies $\tilde{\phi}_i = \phi_i$. 

Now, as \cite{myerson1981optimal} establishes, the seller's revenue can be cast as 
$$\mathbb{E}_{\pmb{v}} \left [ \sum_{i \in [2],k \in [n]} \phi_{i}(v_{i,k}) x_{i,k}(\pmb{v}) \right] = 
\mathcal{R} - \mathcal{E}
$$
with 
$$
\mathcal{R} := \mathbb{E}_{\pmb{v}} \left [ \sum_{i \in [2],k \in [n]} \tilde{\phi}_{i}(v_{i,k}) x_{i,k}(\pmb{v}) \right]$$
and 
$$
\mathcal{E} := \sum_{i \in [2],k \in [n]} \mathbb{E}_{\pmb{v}_{-(i,k)}} \left [ \int_{\underline{v}_i}^{\bar{v}_i} (H_i(F_i(v)) - G_i(F_i(v))) ~dx_{i,k}(v, \pmb{v}_{-(i,k)}) \right ]. 
$$

Note that the seller's problem is to optimize $\mathcal{R} - \mathcal{E}$ subject to the fairness constraint. As \cite{myerson1981optimal} shows, $\mathcal{E}$ is non-negative. Hence, if we maximize $\mathcal{R}$ subject to the fairness constraint and show that, for the optimal allocation, $\mathcal{E}$ happens to be zero, then that allocation is also the maximizer of $\mathcal{R} - \mathcal{E}$ subject to the fairness constraint.

Now, notice that since $\tilde{\phi}_i$ is monotone, our current result gives us the allocation that maximizes $\mathcal{R}$ subject to the fairness constraint, which, as Theorem \ref{thrm:fair_allocation_static} and \ref{theorem:dynamic_both_feasible} suggest, is in the form of $\max_{k} \tilde{\phi}_1(v_{1,k}) - \max_{k} \tilde{\phi}_2(v_{2,k}) \lesseqgtr \gamma$ for some $\gamma$.

Therefore, it suffices to show that for allocations in this form, $\mathcal{E}$ is zero. Notice that $\mathcal{E}$ can only be positive when $H_i(F_i(v)) > G_i(F_i(v))$ and $dx_{i,k}(v, \pmb{v}_{-(i,k)}) > 0$. However, by the definition of the convex hull, when $H_i(F_i(v)) > G_i(F_i(v))$, then $G_i$ is linear in a neighborhood of $v$, which means $g'_i(v) = 0$, implying that $\tilde{\phi}_i(v)$ is constant in a neighborhood of $v$. Now, as the value of others is fixed and the virtual value of buyer $(i,k)$ is also constant in a neighborhood of $v$, given the above form of the allocation, $x_{i,k}(v, \pmb{v}_{-(i,k)})$ also remains constant in that neighborhood, which implies $dx_{i,k}(v, \pmb{v}_{-(i,k)}) = 0$. This completes the proof that $\mathcal{E} = 0$.
\subsection{Proof of \cref{proposition:approx_delta_leq1}}
\label{proof:approximation1}
We first define $\bx'$ as follows: We assume that the auction is designed to run for $T_0$ rounds (for a chosen value of $T_0$ determined later) and use the recursive functions developed earlier to determine the optimal fair allocation under the  ex post fairness constraint at level $\alpha_i - \varepsilon$ for group $i$ (computed over the first $T_0$ rounds). For the remaining $T-T_0$ rounds, we conduct the standard second-price auction.

First, we characterize the approximation $\bx'$ in relation to the exact fair allocation. Let ${{x^*}^t_i}$ denote the allocation to group $i$ at time $t$ under the optimal allocation that satisfies the fairness constraint \eqref{eqn:fairness_constraint} at a level $\alpha_i$ for $i \in \{1, 2\}$. Next, we consider the fairness guarantees of $\bx^* \coloneqq \{ {x^*}^t_i\}_{t=1}^T$ with early stopping. 

\begin{claim}\label{claim:earlystopping_T0}
    If the allocation $\bx^*$ is stopped at time  $T_0 \coloneqq \log(\varepsilon)/\log(\delta) \leq T$,  then $\{{x^*}^t_i\}_{t=1}^{T_0}$ satisfies the ex post fairness constraint at a level $(\alpha_i - \varepsilon)$.
\end{claim}

\begin{proof}
    Observe that, since the allocation $\bx^*$ satisfies the fairness constraint at level $\alpha_i$, the following inequality holds:
\begin{equation}\label{eqn:avg_allocation_fairness}
    \begin{aligned}
        \sum_{t=1}^{T_0}  \delta^{t-1} {x^*}^t_i +  \mathbb{E} \left [ \sum_{t=T_0+1}^T  \delta^{t-1} {x^*}^t_i \right ] \geq \alpha_i \sum_{t=0}^{T-1} \delta^{t}.
    \end{aligned}
\end{equation}
Now, since ${x^*}^t_i \in \{0, 1\}$ for all $t$, we can see that 
\begin{equation*}
    \begin{aligned}
        \sum_{t=T_0+1}^T \delta^{t-1} {x^*}^t_i \leq \delta^{T_0} \sum_{t=0}^{T - 1 - T_0} \delta^{t} &= \varepsilon \sum_{t=0}^{T - 1 - T_0} \delta^{t}
    \end{aligned}
\end{equation*}
where the second equality follows from the fact that $\delta^{T_0} = \varepsilon$ by construction.  From here, we can see that
\begin{equation*}
    \begin{aligned}
    \frac{\sum_{t=1}^{T_0}  \delta^{t-1} {x^*}^t_i }{\sum_{t=0}^{T_0 -1} \delta^t } \geq \alpha_i \frac{\sum_{t=0}^{T-1} \delta^{t}}{{\sum_{t=0}^{T_0 -1} \delta^t }} - \varepsilon \frac{\sum_{t=0}^{T - 1 - T_0} \delta^{t}}{\sum_{t=0}^{T_0 -1} \delta^t } \geq (\alpha_i - \varepsilon) \frac{\sum_{t=0}^{T-1} \delta^{t}}{\sum_{t=0}^{T_0 -1} \delta^t } \geq  \alpha_i - \varepsilon
    \end{aligned}
\end{equation*}
where the last two inequalities follow from the fact that $\delta^t >0$ for all $t$ and that $T_0 \leq T$. It follows that
\begin{equation}\label{eqn:approx_fairness_T0_rounds}
    \begin{aligned}
        \sum_{t=1}^{T_0}  \delta^{t-1} {x^*}^t_i \geq (\alpha_i - \varepsilon) \sum_{t=0}^{T_0 -1} \delta^t. 
    \end{aligned}
\end{equation}
\end{proof}

We now use this claim to show the following claim regarding the fairness guarantee of $\bx'$. 

\begin{claim}
    Over $T$ rounds, $\bx'$ satisfies the fairness condition at a level $(1 - \varepsilon)(\alpha_i - \varepsilon)$ for $i \in \{1, 2\}$.
\end{claim}

\begin{proof}
    Analogously to the proof of the preceding claim, we demonstrate the claim by showing 
     \begin{equation}\label{eqn:fairness_approx_1}
     \begin{aligned}
         \sum_{t=1}^T \delta^{t-1} {x'}^t_i \geq (1 - \varepsilon) (\alpha_i - \varepsilon) \sum_{t=0}^{T-1} \delta^t.
     \end{aligned}
    \end{equation}
    Next, by expanding the left hand side and applying both the result of Claim  \ref{claim:earlystopping_T0} and the observation that, by construction of $\bx'$,  ${x'}^t_i = {x^*}^t_i$ for $t \in [T_0]$, we can see that  
    \begin{equation}\label{eqn:fairness_approx_2}
     \begin{aligned}
         \sum_{t=1}^{T_0} \delta^{t-1} {x'}^t_i + \sum_{t=T_0+1}^{T} \delta^{t-1} {x'}^t_i = \sum_{t=1}^{T_0} \delta^{t-1} {x^*}^t_i + \sum_{t=T_0+1}^{T} \delta^{t-1} {x'}^t_i\geq (\alpha_i - \varepsilon) \sum_{t=0}^{T_0 -1} \delta^t + \sum_{t=T_0+1}^{T} \delta^{t-1} {x'}^t_i
     \end{aligned}
    \end{equation}

    In other words, it suffices to show
    \begin{equation}\label{eqn:fairness_approx_3}
        \begin{aligned}
            (\alpha_i - \varepsilon) \sum_{t=0}^{T_0 -1} \delta^t + \sum_{t=T_0+1}^{T} \delta^{t-1} {x'}^t_i \geq (1 - \varepsilon) (\alpha_i - \varepsilon) \sum_{t=0}^{T-1} \delta^t = 
            (\alpha_i - \varepsilon) \sum_{t=0}^{T-1} \delta^t - \varepsilon (\alpha_i - \varepsilon) \sum_{t=0}^{T-1} \delta^t.
        \end{aligned}
    \end{equation}
    By rearranging, observe that it is equivalent to show that 
    \begin{equation}\label{eqn:fairness_approx_4}
        \begin{aligned}
             \sum_{t=T_0+1}^{T} \delta^{t-1} {x'}^t_i  = 
            (\alpha_i - \varepsilon) \sum_{t=T_0}^{T-1} \delta^t - \varepsilon (\alpha_i - \varepsilon) \sum_{t=0}^{T-1} \delta^t
        \end{aligned}
    \end{equation}
    
    Now, since $\delta \in (0, 1)$ and $T \leq T_0$, we can observe that 
     \begin{equation*}
         \begin{aligned}
             \delta^{T_0} \frac{1 - \delta^T}{1 - \delta} \geq \frac{\delta^{T_0} - \delta^T}{1 - \delta}.
         \end{aligned}
     \end{equation*}
     Substituting $\delta^T_0 = \varepsilon$ and recognizing these expressions as geometric series yields
     \begin{equation*}
         \begin{aligned}
             \varepsilon \sum_{t=0}^{T-1} \delta^t \geq \sum_{t = T_0}^{T-1} \delta^t. 
         \end{aligned}
     \end{equation*}
    Now, noticing that $\alpha_i - \varepsilon \geq 0$ and rearranging, we see that 
    \begin{equation*}
        \begin{aligned}
            0 \geq (\alpha_i - \varepsilon)\sum_{t = T_0}^{T-1} \delta^t - \varepsilon(\alpha_i - \varepsilon) \sum_{t=0}^{T-1} \delta^t.
        \end{aligned}
    \end{equation*}

  Notice that \Cref{eqn:fairness_approx_4} follows by the fact that since $\delta^{t-1}{x'}^t_i \geq 0$ for all $t$. The claim follows.
\end{proof}

Next, we consider the computational guarantee of $x'$. Since the fair allocation algorithm is recursive, its computational complexity increases exponentially in the number of rounds. In the approximation, we conduct exactly $T_0$ rounds of the fair allocation, therefore, the computational complexity can be bounded as follows:
\begin{equation*}
    \begin{aligned}
        \mathcal{O}\left( 2^{T_0}\right) \leq \mathcal{O}\left( e^{T_0}\right)  = \mathcal{O}\left(\varepsilon^{1/\log(\delta)} \right) =\mathcal{O}\left( \frac{1}{\varepsilon}^{\frac{1}{\log(1/\delta)}}\right).
    \end{aligned}
\end{equation*}

Finally, we consider seller utility guarantees of $\bx'$. First, we note that, by Claim \ref{claim:earlystopping_T0}, the seller utility up to time $T_0$ under $\bx'$ is at least that of $\bx^*$ up to $T_0$ because the set of allocations over which $\bx'$ is optimal over the first $T_0$ rounds in terms of seller utility is a subset of those allocations over which $\bx^*$ is optimal. Second, since, under $\bx'$ we perform an unconstrained second-price auction in the remaining $T - T_0$ rounds, it follows that $\bx'$, must also guarantee a seller utility that is at least that of the optimal fair allocation $\bx^*$ over the second interval. Therefore, $\bx'$ guarantees a seller utility that is at least that of $\bx^*$ over all $T$ rounds. $\blacksquare$

\subsection{Proof of \cref{proposition:approx_delta_approx1}}
\label{proof:approximation2}
First, we give the full proposition statement: 
\begin{proposition}\label{proposition:approx_delta_approx1_full}
Suppose Assumptions \ref{assumption:regularity}, \ref{assumption:no_reserve_price} , and \ref{assumption:oracle} hold and that $\delta < 1$. Then, for any $\varepsilon \in (0, \min_i(\alpha_i)), \beta > 1 - \delta$, there exists an approximation to the optimal allocation $\bx''$ with the following properties: 
\begin{enumerate}
    \item $\bx''$ satisfies the fairness constraint over all rounds  at a level of at least $(1-\varepsilon)(1-\beta)^2(\alpha_i - \varepsilon)$.
    \item $\bx''$ guarantees that the seller's total utility is at least $(1 - \beta)$ of that of the optimal allocation.
    \item $\bx''$ can be computed by calling the oracle $\mathcal{O}\left( {\frac{1}{\varepsilon}}^{\frac{1}{\beta} \log\left(\frac{1}{1 - \delta}\right)+1}\right)$ times.
\end{enumerate}
\end{proposition}
Note that $\varepsilon, \beta$ can be chosen to achieve a given $c \leq \delta^2$ such that the proof statement in the body holds.

Now, we proceed with the proof by first defining $x''$ as follows: As in the approximation $\bx'$ presented in \cref{proposition:approx_delta_leq1}, we assume under $\bx''$ that the fair auction runs over $T_0$ rounds (the same value $T_0$ which we detail later) where we partition these $T_0$ rounds into $T_0/\ell$ buckets of $\ell$ rounds (we later also define $\ell$ in detail) and approximate the discount factor as being constant within each bucket. Then, within each bucket, we recursively calculate the fair allocation at level $(1 - \beta)(1-\alpha_i)$ with respect to the approximated discontinuous discounting scheme. For the remaining $T - T_0$ rounds, we conduct a standard second-price auction. 

In particular, as before, we take $T_0  \coloneqq \frac{\log(\varepsilon)}{\log(\delta)}$ and we choose $\ell$ such that $\delta^\ell \simeq 1 - \beta$. Without loss of generality, we assume $T_0/\ell$, $\ell$ are integers. In the $k$th bucket, we approximate the discount factor of each round in the bucket with $\delta^{k-1}$ for $k \in [T_0/\ell]$. 

Now, we characterize the allocation under $\bx''$ in relation to the exact fair allocation $\bx^*$ as defined in the proof of \cref{proposition:approx_delta_leq1} to show its fairness guarantee.  First, observe the following relationship between the true discount factors and those of the discontinuous discounting scheme in $\bx''$:
\begin{equation}\label{eqn:discount_relationship}
    \begin{aligned}
        \sum_{t=0}^{T_0-1} \delta^t \geq (1 - \beta) \sum_{k=1}^{T_0/\ell} \ell \cdot \delta^{\ell (k-1)}.
    \end{aligned}
\end{equation}
Therefore, by Claim \ref{claim:earlystopping_T0} and our construction of the discontinuous discounting scheme, $\bx^*$ satisfies the fairness constraint at the level $(1-\beta)(\alpha_i-\varepsilon)$ 
 with respect to the discontinuous discounting scheme over the first $T_0$ rounds, i.e.,
\begin{equation*}
    \begin{aligned}
        \sum_{k=0}^{T_0/\ell-1}  \delta^{\ell k} \sum_{t = k \ell + 1}^{k \ell} {x_i^*}^t \geq \sum_{t=1}^{T_0} \delta^{t-1} {x_i^*}^t \geq (1 - \beta)(\alpha_i - \varepsilon) \sum_{k=1}^{T_0/\ell} \ell \cdot \delta^{\ell (k-1)}.
    \end{aligned}
\end{equation*}
From here, we characterize $\bx''$ as the allocation such that, in the first $T_0$ rounds, we perform the fair allocation procedure that satisfies fairness constraint at a level $ (1 - \beta) (\alpha_i - \varepsilon)$ with respect to the discontinuous discounting scheme and a standard second-price auction in the remaining $T-T_0$ rounds. It follows that $\bx''$ satisfies
\begin{equation*}
    \begin{aligned}
        \sum_{t=0}^{T_0-1} \delta^t {x_i''}^t \geq (1 - \beta) \sum_{k=0}^{T_0/\ell-1}  \delta^{\ell k} \sum_{t = k \ell + 1}^{k \ell} {x_i''}^t \geq (1 - \beta)^2 (\alpha_i - \varepsilon) \sum_{t=0}^{T_0-1} \delta^t.
    \end{aligned}
\end{equation*}
In other words, over the first $T_0$ rounds, $\bx''$ satisfies the fairness constraint at a level of at least $(1 - \beta)^2 (\alpha_i - \varepsilon)$ with respect to the original discounting scheme. By the same manner as the proof of \cref{proposition:approx_delta_leq1}, it follows that, over $T$ rounds, $\bx''$ guarantees fairness at the level at least $(1-\varepsilon)(1 - \beta)^2(\alpha_i - \varepsilon)$. Finally, by \eqref{eqn:discount_relationship} and the same logic as the proof of \cref{proposition:approx_delta_leq1}, we conclude that the seller's utility is at least $(1 - \beta)$ that of $\bx^*$ under $\bx''$ over $T$ rounds. $\blacksquare$
\subsection{Proof of \cref{theorem:multi_group_static}}
\label{proof:multi_group_static}

The multi-group result follows by generalizing from the proof of \cref{thrm:fair_allocation_static} in \cref{proof:static_alloc}. 

First, fix $i, j \in [L]$ such that $i \neq j$. Let $\mathcal{S}$ be the set of $(\phi_{i}(v_{i,k}))_{i,k}$ for which the seller allocates the item to buyers in either group $i$ or group $j$ in the optimal allocation $\tilde{x}$, i.e.,
\begin{equation}\label{eqn:S_multi}
\mathcal{S} = \left \{ (\phi_{i}(v_{i,k}))_{i,k} \Big \vert
\exists k \text{ s.t. } \tilde{x}_{i,k}(\bv) = 1 \text{ or } \exists k \text{ s.t. } \tilde{x}_{j,k}(\bv) = 1
\right \}.
\end{equation}

In the following lemma, we see that the allocation rule that defines the allocation boundary between groups $i$ and $j$ in the space of virtual values is an \emph{affine allocation}, as defined in \cref{proof:static_alloc}.

\begin{lemma}
     For any $i$ and $j \in [L]$ , there exists $\eta_{i,j}$ such that, if the item is allocated to one of the groups $i$ or $j$, then it is allocated to group $i$ if and only if $\phi_i(v_i) - \phi_j(v_j) \geq \eta_{i,j}$
\end{lemma}

This lemma follows from the proof of \cref{proof:static_alloc} where we let, without loss of generality, groups $i$ and $j$ play the role of groups 1 and 2, respectively and use the definition of $\mathcal{S}$ given by \cref{eqn:S_multi}.

Next, in the following lemma, we establish the relationship of the parameters that define the optimal allocation. 

\begin{lemma}
    Suppose, without loss of generality, that $\eta_{i,j} \geq 0$ and $\eta_{j,k} \geq 0$. Then, we have $\eta_{i,k} = \eta_{i,j} + \eta_{j,k}$.
\end{lemma}

This result follows by a similar argument as is made in  \cref{proof:static_alloc}.
\subsection{Proof of Theorem \ref{theorem:dynamic_multi}} 
\label{proof:multi_group_dynamic}

 This proof extends the arguments used in the proof of \cref{theorem:dynamic_both_feasible}. Recall that the variable
\begin{equation}
\sum_{\ell=1}^n x_{i,\ell}^t(\bv^t).    
\end{equation}
indicates whether group $i$ wins the item in round $t$. By Assumption \ref{assumption:no_reserve_price}, notice that, for any $i \in [L], k \in [n]$, we can write the utility of buyer $(i,k)$ as 
\begin{equation}\label{eqn:dynamic_multi_proof_1}
    x_{i,k}^t(b, \bv_{-(i,k)}^t) v_{i,k}^t - p_{i,k}^t(b, \bv_{-(i,k)}^t) + \delta \sum_{j=1}^L \sum_{\ell = 1}^n x_{j, \ell}^t\left( b , \bv_{-(i,k)}^t\right) \nu_i^{t+1}\left( (R_{j}^t-1)/\delta, R_{-j}^t/\delta \right).
\end{equation}
We can again write the adjusted utility as 
\begin{equation} 
x_{i,k}^t(b, \bv_{-(i,k)}^t) v_{i,k}^t - \tilde{p}_{i,k}^t(b, \bv_{-(i,k)}^t),
\end{equation}
where 
\begin{equation}
\tilde{p}_{i,k}^t(b, \bv_{-(i,k)}^t) := p_{i,k}^t(b, \bv_{-(i,k)}^t)
- \delta \sum_{j=1}^L \sum_{\ell = 1}^n x_{j, \ell}^t\left( b , \bv_{-(i,k)}^t\right) \nu_i^{t+1}\left( (R_{j}^t-1)/\delta, R_{-j}^t/\delta \right).   
\end{equation}

By the same argument as the proof of \cref{theorem:dynamic_both_feasible}, 
\begin{equation}\label{eqn:dynamic_multi_proof_p}
    \begin{aligned}
        & p_{i,k}^t(v, \bv_{-(i,k)}^t) = v~x_{i,k}^t(v, \bv_{-(i,k)}^t)- \\ & \int_{\underline{v}_i^t}^v x_{i,k}^t(z, \bv_{-(i,k)}^t) dz + \delta \sum_{j=1}^L \sum_{\ell = 1}^n x_{j, \ell}^t\left( v , \bv_{-(i,k)}^t\right) \nu_i^{t+1}\left( (R_{j}^t-1)/\delta, R_{-j}^t/\delta \right) + c_{i,k}(\bv_{-(i,k)}^t),
    \end{aligned}
\end{equation}
where $c_{i,k}(\cdot)$ is some function that is independent of $v$. By the same argument as the proof of \cref{theorem:dynamic_both_feasible}, the expected payment of buyer $(i,k)$ to the seller at round $t$ is therefore given by
\begin{equation}
\mathbb{E}_{\bv^t} \left[ \phi_i^t(v_{i,k}^t) x_{i,k}^t(\bv^t) 
+ \delta \sum_{j=1}^L \sum_{\ell = 1}^n x_{j, \ell}^t\left(\bv^t\right) \nu_i^{t+1}\left( (R_{j}^t-1)/\delta, R_{-j}^t/\delta \right),
\right ] + c_i^t,
\end{equation}
where $c_i^t$ denotes the group-wise value for $c_{i,k}(\cdot)$ that satisfies the IR constraint with equality. The seller's utility at round $t$ is, therefore, 
\begin{equation}
    \begin{aligned}
        \mathbb{E}_{\bv^t} \left[ \sum_{i=1}^L \sum_{k=1}^n \phi_i^t(v_{i,k}^t) x_{i,k}^t(\bv^t) 
+  n \delta \sum_{j=1}^L \sum_{\ell=1}^n x_{i,\ell}^t(\bv^t) \left( \sum_{i=1}^L \nu_i^{t+1}\left( (R_{j}^t-1)/\delta, R_{-j}^t/\delta \right)\right) \right ]  + n \sum_{i=1}^L c_i^t.
    \end{aligned}
\end{equation}
It follows that the seller's utility from time $t$ onwards is
\begin{align}
&\mathbb{E}_{\bv^t} \left[ \sum_{i=1}^L \sum_{k=1}^n \phi_i^t(v_{i,k}^t) x_{i,k}^t(\bv^t) + \right. \nonumber \\
& \left.
\delta \sum_{j=1}^L \sum_{\ell=1}^n x_{j,\ell}^t(\bv^t) \left( n \sum_{i=1}^L \nu_i^{t+1}\left( (R_{j}^t-1)/\delta, R_{-j}^t/\delta \right) + \mu^{t+1}\left((R_{j}^t-1)/\delta, R_{-j}^t/\delta\right) \right) \right ]  + n \sum_{i=1}^L c_i^t.
\end{align}
As in the proof of \cref{theorem:dynamic_both_feasible}, it follows that we allocate to group $i \in [L]$ in round $t$ in the following region $G_i^t$:
\begin{equation}\label{eqn:G_i_t_multi}
    \begin{aligned}
        \hspace{-3cm} G_i^t := \Big \{ (v_1^t, \hdots, v_L^t) ~\Big \vert~   \phi_i^t(v_i) - \phi_{j}^t(v_{j}) \geq  \delta n \sum_{\ell = 1}^L \Delta_{\ell}^t(j, i) + \delta \Delta_0^{t}(j, i) \: \forall j \in [L], j \neq i \Big \}.
    \end{aligned}
\end{equation}

\end{document}